\documentclass[sn-mathphys-num,pdflatex]{sn-jnl}

\usepackage{fix-cm}
\usepackage{graphicx}
\usepackage{multirow}
\usepackage{amsmath,amssymb,amsfonts}
\usepackage{amsthm}
\usepackage{mathrsfs}
\usepackage[title]{appendix}
\usepackage{xcolor}
\usepackage{textcomp}
\usepackage{manyfoot}
\usepackage{booktabs}
\usepackage{listings}
\usepackage{tabularray}
\usepackage{overpic}
\usepackage{mathtools}
\usepackage{comment}
\usepackage[linesnumbered,ruled,vlined]{algorithm2e}
\usepackage{enumitem}
\usepackage{longtable}
\usepackage{array}
\usepackage{placeins}

\renewcommand{\le}{\leqslant}

\renewcommand{\ge}{\geqslant}

\newcommand{\cC}{\mathcal{C}}

\newcommand{\eqdef}{\triangleq}

\newcommand{\B}{\mathcal{B}}

\newcommand{\LPC}{\operatorname{LPC}_{\infty}}
\newcommand{\dH}{d_{\mathrm H}}
\newcommand{\Dmin}{D_{\min}}
\newcommand{\Tball}{\mathsf{T}}

\newcommand{\wt}{\operatorname{wt}}

\newcommand{\Tr}{\operatorname{Tr}}
\newcommand{\Pref}{\operatorname{Pref}}
\newcommand{\prefx}{\operatorname{pref}}

\theoremstyle{thmstyleone}
\newtheorem{theorem}{Theorem}
\newtheorem{proposition}[theorem]{Proposition}
\newtheorem{corollary}[theorem]{Corollary}
\newtheorem{lemma}[theorem]{Lemma}

\theoremstyle{thmstyletwo}
\newtheorem{example}{Example}
\newtheorem{remark}{Remark}

\theoremstyle{thmstylethree}
\newtheorem{definition}{Definition}
\hypersetup{bookmarksdepth=subsubsection}

\newif\ifshowproofs
\showproofstrue 

\begin{document}

\title[A \(q\)-ary Local Criterion for \(\LPC(1)\)]%
{A \(q\)-ary Local Criterion for the Radius-One Limited Permutation Channel
and Almost-Optimal Binary Block-Concatenation Codes}

\author[1]{\fnm{Noam} \sur{Ben Shimon}}
\email{noamben@campus.technion.ac.il}

\author[2]{\fnm{Aryeh~Lev} \sur{Zabokritskiy (Yohananov)}}
\email{yuhanalev@telhai.ac.il}

\affil[1]{%
\orgdiv{Department of Computer Science},
\orgname{Technion -- Israel Institute of Technology},
\orgaddress{\city{Haifa}, \country{Israel}}
}

\affil[2]{%
\orgdiv{Department of Computer Science},
\orgname{MIGAL -- Galilee Research Institute/Tel-Hai University of Kiryat Shmona and the Galilee},
\orgaddress{\city{Kiryat Shmona}, \country{Israel}}
}

\abstract{
The radius-one limited permutation channel \(\LPC(1)\) maps a transmitted word
to any word obtained by an arbitrary set of pairwise disjoint adjacent
transpositions.  This is the \(r=1\) case of the \(\ell_\infty\)-limited
permutation channel introduced by Langberg et al., and is also the zero-error
version of simultaneous adjacent-swap errors.  We study zero-error
block-concatenation codes for this channel.

Our first contribution is a \(q\)-ary two-stage local criterion for certifying
free block-concatenation codes.  The criterion replaces the global all-length
confusability problem by finitely many local checks between blocks: a
same-length truncated-ball test and a second-stage prefix test for unequal
lengths.  In the binary case, we use this criterion to construct explicit
block-concatenation codes of rates \(0.649872\), \(0.652018\), and
\(0.653618\).  The best construction improves the previous
string-concatenation rate \(0.642805\) and comes within \(0.013049\) of the
known upper bound \(2/3\).

Although the local criterion is only sufficient, we prove that it is
rate-complete: for every alphabet size \(q\), the supremum of \(q\)-ary block
rates certified by the criterion is exactly the \(q\)-ary zero-error capacity
\(C_0^{(q)}\).  Thus the criterion may reject some valid finite block sets, but
it imposes no asymptotic rate loss.  We also give an exact product-automaton
verifier which decides, for a fixed prefix-free binary block set, whether the
induced finite-length codes are correcting for all lengths.

Finally, we study the related detection problem, motivated by settings with a
feedback channel where the receiver may request retransmission after detecting
an error.  We prove a \(q\)-ary pairing upper bound and give a \(q\)-ary local
detecting criterion.  In the binary case, we construct a detecting
block-concatenation code of rate \(0.756707\), compared with the upper bound
\(\frac12\log_2 3\approx0.792481\).
}

\keywords{limited permutation channel, adjacent transpositions, zero-error codes, block-concatenation codes, finite automata, error detection}

\pacs[MSC Classification]{94B25, 94B50, 94B65, 68P30, 68R05}

\maketitle

\section{Introduction}
\label{sec:intro}

Permutation and reordering channels model communication scenarios in which the
symbols of a transmitted word may arrive in a different order from the one in
which they were sent.  Such effects arise naturally in packet networks,
permutation channels, timing channels, molecular communication, and storage
systems affected by synchronization or bit-shift errors; see, for example,
\cite{WalshWeber2008,KovacevicVukobratovic2013,AnantharamVerdu1996,KovacevicPopovski2014,KadloorAdveEckford2012,ShamaiZehavi1991,Krachkovsky1994}.
For background on permutation metrics, see also \cite{DezaHuang1998}.
Langberg et al.~\cite{LangbergSchwartzYaakobi2015,LangbergSchwartzYaakobi2017}
introduced the \(\ell_\infty\)-limited permutation channel, where the channel may
apply a permutation to the transmitted word, but no symbol may move by more than
\(r\) positions.  They studied this channel in the zero-error setting, gave
direct and recursive constructions, and proved upper bounds on achievable
rates.

This paper focuses on the radius-one case.  In \(\LPC(1)\), every admissible
error is a set of pairwise disjoint adjacent transpositions.  Thus, over a
\(q\)-ary alphabet, a transmitted word \(c=c_1c_2\cdots c_n\in\Sigma_q^n\) may be
received as any word obtained from \(c\) by swapping any collection of
non-overlapping adjacent pairs.  A code is zero-error correcting if the error
balls around distinct codewords are disjoint.

The binary case \(q=2\) is the most basic nontrivial case and remains open.
Langberg et al.~\cite{LangbergSchwartzYaakobi2017} proved a general \(q\)-ary
upper bound which, for \(q=2\) and \(r=1\), gives the binary capacity upper bound
\(2/3\).  Their recursive construction gave binary rate \(0.609\), and
Chee et al.~\cite{CheeKiahLingNguyenVuZhang2016} later improved this to a
string-concatenation construction of rate approximately \(0.642805\).  One of the
main contributions of this paper is to improve this explicit binary
construction: our best certified block-concatenation code has rate \(0.653618\),
leaving an absolute gap of only \(0.013049\) from the upper bound \(2/3\).
Related models based directly on adjacent transpositions have also been studied;
in particular, Kova\v{c}evi\'c et al.~\cite{KovacevicGoyalKiah2026} derive
asymptotic bounds for \(q\)-ary strings under simultaneous disjoint
transpositions.

Our goal is to push the binary zero-error rate closer to the upper bound \(2/3\)
while keeping the correctness proof transparent.  We work mainly with
block-concatenation codes.  Given a finite block set
\(P\subseteq\Sigma_q^*\), the length-\(n\) code \(C_n(P)\) consists of all words
in \(P^*\) of length \(n\).  This framework is natural and has been used in
previous constructions, but it creates a boundary problem: even if individual
blocks look safe, an adjacent transposition may cross a block boundary.
Therefore a global all-length proof cannot be replaced by ordinary
ball-disjointness of isolated blocks.

To handle this, we introduce a \(q\)-ary two-stage local criterion.  The criterion
checks same-length block pairs using truncated balls, and unequal-length pairs
using a finite second-stage prefix test.  When the criterion holds, it certifies
that \(C_n(P)\) is correcting for every length \(n\).  The proof is recursive in
the spirit of Chee et al., but the criterion is stated directly in terms of
local boundary behavior rather than through a particular recursive family.

A natural concern is that such a local sufficient criterion might be too
restrictive to approach the true zero-error capacity.  We show that this is not
the case.  Although the criterion is not necessary and may reject some valid
block sets, it is rate-complete: for every alphabet size \(q\), the supremum of
rates certified by the criterion is exactly the \(q\)-ary zero-error capacity
\(C_0^{(q)}\).  Thus the criterion itself is not an asymptotic bottleneck; any
remaining rate gap in our explicit finite constructions is a limitation of the
particular block families found, not of the certification method.

We also give an exact product-automaton verifier.  The verifier is presented for
the binary case and decides the all-length confusability problem for a fixed
prefix-free block set.  Thus the paper has two complementary certification
tools: the criterion gives short mathematical certificates and supports symbolic
reasoning, while the verifier gives exact validation for finite block sets.

Finally, as in classical coding theory, we also study error detection.  Detection
is operationally meaningful in settings with feedback: the receiver may only
need to recognize that an error occurred and then request retransmission from the
sender.  Although detection is weaker than correction, for \(\LPC(1)\) it leads
to a genuinely different directed problem.  For correction, the obstruction is
the symmetric condition \(\B(c)\cap\B(c')\ne\emptyset\); for detection, the
obstruction is that a nontrivial channel output is itself another codeword,
namely \(c'\in\B(c)\).  We prove a simple \(q\)-ary pairing upper bound, give a
\(q\)-ary local detecting criterion, and construct a binary detecting
block-concatenation code of rate \(0.756707\).

\subsection*{Our contributions}

The main contributions of this paper are as follows.  We distinguish throughout
between results that are stated for a general \(q\)-ary alphabet and results
where the new explicit constructions are binary.

\begin{itemize}[leftmargin=*]
    \item We formulate a \(q\)-ary two-stage local criterion for
    block-concatenation codes in \(\LPC(1)\).  The criterion is sufficient,
    finite, and applies simultaneously to all lengths.  This part is not
    restricted to the binary alphabet.

    \item We apply the criterion to the binary case and construct explicit
    binary block-concatenation codes of rates \(0.649872\), \(0.652018\), and
    \(0.653618\).  The best construction achieves about \(98.04\%\) of the known
    binary upper bound \(2/3\), and improves the previous binary
    string-concatenation rate \(0.642805\).

    \item We prove that the \(q\)-ary local criterion is rate-complete for every
    alphabet size: although it is only sufficient, the supremum of \(q\)-ary
    block rates certified by the criterion is exactly the \(q\)-ary zero-error
    capacity \(C_0^{(q)}\).  Hence the criterion itself is not the asymptotic
    bottleneck.

    \item We give an exact product-automaton verifier for binary block sets.
    For a fixed prefix-free binary block set \(P\), the verifier decides whether
    all finite-length codes \(C_n(P)\) are correcting.  The automaton idea is
    finite-state in nature and can be generalized to larger alphabets, but the
    verifier used in this paper is presented and implemented in the binary case.

    \item We study error detection for \(\LPC(1)\).  The pairing upper bound and
    the local detecting criterion are \(q\)-ary, while the new explicit detecting
    construction is binary.  In the binary case, we obtain a detecting
    construction of rate \(0.756707\), compared with the pairing upper bound
    \[
        \frac12\log_2 3\approx0.792481 .
    \]
\end{itemize}

The rest of the paper is organized as follows.  Section~\ref{sec:prelim}
introduces the channel model, block-concatenation notation, and rate
calculation.  Section~\ref{sec:criterion} proves the \(q\)-ary two-stage local
criterion.  Section~\ref{sec:construction} presents the binary correcting
constructions and their verification tables.  Section~\ref{sec:criterion-properties}
studies structural properties of the criterion, including rate-completeness.
Section~\ref{sec:verifier} describes the exact product-automaton verifier and
compares it with the local criterion.  Section~\ref{sec:detect} treats error
detection.  We conclude with open problems and discussion.

\section{Preliminaries}
\label{sec:prelim}

Let \(\Sigma_q=\{0,1,\dots,q-1\}\) be a finite alphabet, and let
\([n]=\{1,2,\dots,n\}\).  Let \(S_n\) denote the set of all permutations over
\([n]\).  We write a permutation \(\pi\in S_n\) in vector notation as
\[
    \pi=[\pi(1),\pi(2),\dots,\pi(n)],
\]
and denote the identity permutation by
\[
    Id=[1,2,\dots,n].
\]
For two permutations \(\pi,\sigma\in S_n\), define their \(\ell_\infty\)-distance by
\[
    d_\infty(\pi,\sigma)
    =
    \max_{i\in[n]}\lvert \pi(i)-\sigma(i)\rvert .
\]
The \(\ell_\infty\)-weight of \(\pi\) is
\[
    \wt(\pi)
    =
    d_\infty(\pi,Id)
    =
    \max_{i\in[n]}\lvert \pi(i)-i\rvert .
\]

For a word \(c=c_1c_2\cdots c_n\in\Sigma_q^n\) and a permutation
\(\pi\in S_n\), we write \(\pi c\) for the word \(z=z_1z_2\cdots z_n\) defined by
\[
    z_i=c_{\pi(i)},
    \qquad i\in[n].
\]

\begin{definition}[The limited permutation channel]
\label{def:lpc-channel}
The radius-\(r\) \(\ell_\infty\)-limited permutation channel, denoted
\(\LPC(r)\), maps a transmitted word \(c\in\Sigma_q^n\) to any received word
\(z\in\Sigma_q^n\) of the form
\[
    z=\pi c
\]
for some permutation \(\pi\in S_n\) satisfying \(\wt(\pi)\le r\).  Equivalently,
no transmitted symbol is displaced by more than \(r\) positions.
\end{definition}

\begin{definition}[Error ball]
\label{def:lpc-ball}
For \(c\in\Sigma_q^n\), the radius-\(r\) ball centered at \(c\) is
\[
    B_r(c)
    =
    \{z\in\Sigma_q^n:z=\pi c,\ \pi\in S_n,\ \wt(\pi)\le r\}.
\]
Thus, if \(c\) is transmitted over \(\LPC(r)\), the received word may be any
element of \(B_r(c)\).
\end{definition}

\begin{definition}[\(\LPC\)-distance]
\label{def:lpc-distance}
For \(c,c'\in\Sigma_q^n\), define
\[
    d_{\LPC}(c,c')
    =
    \min\{\wt(\pi):\pi\in S_n,\ c'=\pi c\}.
\]
If no such permutation exists, we set \(d_{\LPC}(c,c')=\infty\).
\end{definition}

In particular, \(d_{\LPC}(c,c')<\infty\) only if \(c\) and \(c'\) have the same
symbol composition.  Although this distance is useful for comparison with the
standard formulation of the limited permutation channel, most of the arguments
below are phrased directly in terms of error balls.

\begin{definition}[Confusability and correcting codes]
\label{def:lpc-code}
Two words \(c,c'\in\Sigma_q^n\) are \emph{confusable} in \(\LPC(r)\) if
\[
    B_r(c)\cap B_r(c')\ne\emptyset.
\]
A nonempty code \(\cC\subseteq\Sigma_q^n\) is an \((n,M;r)_q\)-\(\LPC\) code if
\(\lvert\cC\rvert=M\) and no two distinct codewords are confusable, equivalently
\[
    B_r(c)\cap B_r(c')=\emptyset
    \qquad
    \text{for all distinct }c,c'\in\cC.
\]
\end{definition}

\begin{definition}[Optimal code size]
\label{def:Aq}
For integers \(n,q\ge2\) and \(r\ge1\), let \(A_q(n;r)\) denote the largest
integer \(M\) for which an \((n,M;r)_q\)-\(\LPC\) code exists.
\end{definition}

This paper focuses on the radius-one case \(r=1\), and the new explicit
constructions are binary.  When \(r=1\), every admissible permutation is a
product of pairwise disjoint adjacent transpositions.  We write
\[
    \B(c)\eqdef B_1(c).
\]
Thus \(\B(c)\) is the set of all words obtained from \(c\) by swapping any set of
pairwise disjoint adjacent pairs.  In binary sections we abbreviate
\(\Sigma\eqdef\Sigma_2=\{0,1\}\).

\begin{example}[A radius-one error]
For \(c=0110\), the effective swaps may occur at the two transition boundaries
\((1,2)\) and \((3,4)\), and these swaps are disjoint.  Hence
\[
    \B(0110)=\{0110,\ 1010,\ 0101,\ 1001\}.
\]
\end{example}

\begin{definition}[Detection]
\label{def:detection}
A code \(\cC\subseteq\Sigma_q^n\) is \emph{detecting} for \(\LPC(1)\) if
\(\B(c)\cap\cC=\{c\}\) for every \(c\in\cC\).  Every correcting code is
detecting, but not conversely.
\end{definition}

Since the channel only permutes coordinates, in the binary case it preserves
Hamming weight.  Thus words of different Hamming weights are automatically not
confusable.  For a binary word \(c\), we write \(\wt(c)\) for its Hamming weight.
For \(k\ge0\), \(\prefx_k(c)\) denotes the length-\(k\) prefix of \(c\), with
\(\prefx_0(c)=\epsilon\).

The zero-error capacity of the radius-one channel over a \(q\)-ary alphabet is
\[
    C_0^{(q)}=\limsup_{n\to\infty}\frac{1}{n}\log_q A_q(n;1).
\]
and we abbreviate the binary capacity by \(C_0\eqdef C_0^{(2)}\). 
Langberg et al.~\cite{LangbergSchwartzYaakobi2017} proved that for every
\(q\ge2\) and every \(n\) divisible by \(3\),
\[
    A_q(n;1)\le
    \left(q+2\binom{q}{2}+2\binom{q}{3}\right)^{n/3}.
\]
Substituting \(q=2\) gives
\[
    q+2\binom{q}{2}+2\binom{q}{3}
    =2+2\binom{2}{2}+2\binom{2}{3}=4.
\]
Hence, for \(3\mid n\), \(A_2(n;1)\le4^{n/3}=2^{2n/3}\), and therefore
\(C_0\le2/3\).

\subsection{Block-concatenation codes and rates}
 
A natural way to construct codes for \(\LPC(1)\) is by concatenating words from
a finite set of carefully chosen blocks.  In such a construction, one first
selects a finite block set \(P\subseteq\Sigma_q^*\), and then forms codewords by
free concatenation of blocks from \(P\).  This approach is particularly useful
because the asymptotic rate is determined only by the length profile of \(P\),
while correctness can often be certified by checking finitely many local
interactions between blocks.  This is the framework used in the block and
string-concatenation constructions of Langberg et
al.~\cite{LangbergSchwartzYaakobi2017} and Chee et
al.~\cite{CheeKiahLingNguyenVuZhang2016}, and it is also the framework used for
the explicit constructions in this paper.  We note, however, that
block-concatenation is not the only possible construction paradigm; direct,
recursive, computer-aided, and more general constrained or finite-state
constructions are also natural in this setting.

\begin{definition}[Block-concatenation code]
\label{def:block-code}
Let \(P\subseteq\Sigma_q^*\) be a finite set of nonempty blocks.  We write
\[
    P^*
    =
    \{p_1p_2\cdots p_m:m\ge0,\ p_i\in P\}
\]
for the free concatenation language generated by \(P\), where the case \(m=0\)
gives the empty word \(\epsilon\).  For \(n\ge0\), define
\[
    C_n(P)
    =
    P^*\cap\Sigma_q^n.
\]
Equivalently,
\[
    C_n(P)
    =
    \{p_1p_2\cdots p_m\in\Sigma_q^n:m\ge0,\ p_i\in P,
        \sum_{i=1}^{m}\lvert p_i\rvert=n\}.
\]
For \(d\ge0\), let \(\Pref_d(P^*)\) be the set of length-\(d\) prefixes of words
in \(P^*\), with \(\Pref_0(P^*)=\{\epsilon\}\).
\end{definition}

\begin{definition}[Length profile and block rate]
\label{def:block-rate}
Let \(P\subseteq\Sigma_q^*\) be a finite nonempty block set.  Its length profile
is \((p_\ell)_{\ell\ge1}\), where
\[
    p_\ell=\lvert\{p\in P:\lvert p\rvert=\ell\}\rvert .
\]
The associated growth parameter \(\lambda\) is defined as the unique positive
solution of
\[
    \sum_{\ell\ge1}p_\ell\lambda^{-\ell}=1.
\]
The corresponding profile rate over the \(q\)-ary alphabet is
\[
    R(P)=\log_q\lambda .
\]
\end{definition}

The number \(R(P)\) depends only on the length profile of \(P\).  It is obtained
by counting formal concatenations of blocks, as in earlier block and
string-concatenation constructions for the limited permutation channel
\cite{LangbergSchwartzYaakobi2017,CheeKiahLingNguyenVuZhang2016}.  Indeed, if
\[
    A_P(z)=\sum_{\ell\ge1}p_\ell z^\ell
\]
is the one-block length generating function, then \(A_P(z)^m\) counts formal
concatenations of exactly \(m\) blocks.  Hence formal concatenations of any
number of blocks are counted by
\[
    \sum_{m\ge0}A_P(z)^m
    =
    \frac{1}{1-A_P(z)}
    =
    \frac{1}{1-\sum_{\ell\ge1}p_\ell z^\ell}.
\]
The smallest positive singularity \(\rho\) of this rational function satisfies
\(A_P(\rho)=1\).  Writing \(\lambda=\rho^{-1}\), this is equivalent to
\[
    \sum_{\ell\ge1}p_\ell\lambda^{-\ell}=1.
\]
By the standard coefficient asymptotics for rational generating functions with
nonnegative coefficients, the number of formal concatenations of total length
\(n\) grows as \(\lambda^{n+o(n)}\) along the admissible lengths; see, e.g.,
\cite[Ch.~IV]{FlajoletSedgewick2009}.

If \(P\) is prefix-free, then different formal concatenations give different
words.  Indeed, if
\[
    p_1p_2\cdots p_m=q_1q_2\cdots q_s,
    \qquad p_i,q_j\in P,
\]
then the first blocks \(p_1\) and \(q_1\) are prefixes of the same word, so one
is a prefix of the other.  Since \(P\) is prefix-free, \(p_1=q_1\).  Cancelling
and repeating gives \(m=s\) and \(p_i=q_i\) for all \(i\).  Therefore, in the
prefix-free case,
\[
    \sum_{n\ge0}\lvert C_n(P)\rvert z^n
    =
    \sum_{m\ge0}A_P(z)^m
    =
    \frac{1}{1-A_P(z)}.
\]
Consequently, along the admissible lengths,
\[
    \lvert C_n(P)\rvert=\lambda^{n+o(n)},
\]
and the profile rate \(R(P)\) is the actual asymptotic rate of the
finite-length codes \(C_n(P)\).  For block sets that are not prefix-free, the
same formula still gives the profile rate, but the actual number of distinct
words in \(C_n(P)\) may be smaller because different block decompositions may
produce the same word.

\begin{example}[A small block-concatenation code]
Let \(P=\{00,11\}\).  Then \(P^*\) consists of all binary words obtained by
repeating each source symbol twice, for example
\[
    \epsilon,\quad 00,\quad 11,\quad 0000,\quad 0011,\quad 1100,\quad 1111,\quad \dots .
\]
For every \(m\ge0\), the length-\(2m\) code \(C_{2m}(P)\) has exactly \(2^m\)
words, and \(C_n(P)=\emptyset\) for odd \(n\).  The length profile is \(p_2=2\),
so the rate equation is
\[
    2\lambda^{-2}=1.
\]
Hence \(\lambda=\sqrt2\), and
\[
    R(P)=\log_2\sqrt2=\frac12.
\]
\end{example}

More generally, over a \(q\)-ary alphabet, the repeated-symbol block set
\[
    P_{\mathrm{rep},q}=\{aa:a\in\Sigma_q\}
\]
has length profile \(p_2=q\).  Thus \(q\lambda^{-2}=1\), so
\(\lambda=\sqrt q\), and
\[
    R(P_{\mathrm{rep},q})=\log_q\sqrt q=\frac12.
\]
This gives the trivial rate-\(1/2\) benchmark.  We do not prove its correctness
here; it follows immediately from the local criterion developed below.

\section{The two-stage local criterion}
\label{sec:criterion}

A direct proof that the finite-length block-concatenation code \(C_n(P)\) is
correcting for every \(n\) would require checking pairs of distinct concatenations
of total length \(n\), with no a priori bound on the number of blocks.  The
purpose of the criterion below is to replace this global condition by finitely
many local checks on pairs of blocks.  The criterion is \(q\)-ary and applies to
any finite block set \(P\subseteq\Sigma_q^*\).  It is only sufficient, not
necessary, but when it holds it gives a short certificate that \(C_n(P)\) is
correcting for every length \(n\).  Moreover, as we prove in
Section~\ref{sec:criterion-properties}, this sufficiency does not impose an
asymptotic rate loss: the supremum of rates certified by the criterion is the
zero-error capacity itself.  Thus the criterion may reject some valid block
sets, but it is not an asymptotic bottleneck.  Its exact counterpart is the
product-automaton verifier of Section~\ref{sec:verifier}.

For readability, we sometimes insert a separator \(u\,|\,v\) when displaying a
concatenation, but the separator is not part of the word.  For equal-length
words \(u,v\in\Sigma_q^\ell\), define the minimum Hamming distance between their
error balls by
\[
    \Dmin(u,v)=\min\{\dH(u',v'):u'\in\B(u),\ v'\in\B(v)\}.
\]
When \(\Dmin(u,v)=1\), a pair \((u',v')\in\B(u)\times\B(v)\) with
\(\dH(u',v')=1\) is called a \emph{closest witness}.

\begin{definition}[Truncated ball]
\label{def:truncated-ball}
For \(w\in\Sigma_q^\ell\), define the \emph{truncated ball}
\[
    \Tball(w)=\prefx_{\ell-1}(\B(w))
    =
    \{\prefx_{\ell-1}(w'):w'\in\B(w)\}\subseteq\Sigma_q^{\ell-1}.
\]
\end{definition}

The truncated ball records all possible received prefixes before the right
boundary of the word.  This is the relevant object for block concatenation:
after a block is followed by a continuation, a radius-one transposition can
affect the interface only through the last coordinate of the block.

\begin{definition}[Same-length boundary test]
\label{def:samelen}
Let \(u,v\in\Sigma_q^\ell\).  The ordered pair \((u,v)\) passes the same-length
boundary test if:
\begin{enumerate}[label=\textup{(E\arabic*)}]
    \item \(\Dmin(u,v)\ne0\), equivalently \(\B(u)\cap\B(v)=\emptyset\);
    \item if \(\Dmin(u,v)=1\), then no closest witness differs in the last
    coordinate \(\ell\).
\end{enumerate}
If \(\Dmin(u,v)\ge2\), condition \textup{(E2)} is vacuous.
\end{definition}

\begin{lemma}[Truncation form]
\label{lem:trunc}
For distinct \(u,v\in\Sigma_q^\ell\), the pair \((u,v)\) passes the
same-length boundary test if and only if
\[
    \Tball(u)\cap\Tball(v)=\emptyset.
\]
Moreover, truncation is injective on \(\B(w)\), and hence
\[
    \lvert\Tball(w)\rvert=\lvert\B(w)\rvert .
\]
\end{lemma}

\ifshowproofs
\begin{proof}
We first prove the injectivity statement.  All words in \(\B(w)\) have the same
symbol composition as \(w\), since the channel only permutes coordinates.  Hence,
if two elements of \(\B(w)\) have the same prefix of length \(\ell-1\), then
their last symbols are forced to be equal as well.  Therefore the two words are
identical.  Thus truncation is injective on \(\B(w)\), and consequently
\(\lvert\Tball(w)\rvert=\lvert\B(w)\rvert\).

It remains to prove the equivalence between the same-length boundary test and
disjointness of truncated balls.  We prove both directions.

First assume that \((u,v)\) passes the same-length boundary test.  We show that
\(\Tball(u)\cap\Tball(v)=\emptyset\).  Suppose, to the contrary, that
\(z\in\Tball(u)\cap\Tball(v)\).  Then there exist symbols
\(\alpha,\beta\in\Sigma_q\) such that
\[
    z\alpha\in\B(u),
    \qquad
    z\beta\in\B(v).
\]
If \(\alpha=\beta\), then \(z\alpha\in\B(u)\cap\B(v)\), contradicting
\textup{(E1)}.  If \(\alpha\ne\beta\), then \(z\alpha\) and \(z\beta\) differ
exactly in the last coordinate \(\ell\).  Since \textup{(E1)} excludes
\(\Dmin(u,v)=0\), this implies \(\Dmin(u,v)=1\), and the pair
\((z\alpha,z\beta)\) is a closest witness whose unique disagreement is in the
last coordinate.  This contradicts \textup{(E2)}.  Therefore the truncated balls
must be disjoint.

Conversely, assume that \(\Tball(u)\cap\Tball(v)=\emptyset\).  We show that the
same-length boundary test holds.  If \textup{(E1)} failed, then there would be a
word \(z'\in\B(u)\cap\B(v)\).  Its prefix \(\prefx_{\ell-1}(z')\) would then
belong to both \(\Tball(u)\) and \(\Tball(v)\), contradicting the assumed
disjointness.  Thus \textup{(E1)} holds.

Now suppose that \(\Dmin(u,v)=1\).  If there were a closest witness
\(u'\in\B(u)\), \(v'\in\B(v)\) differing in the last coordinate \(\ell\), then
\(u'\) and \(v'\) would have the same prefix of length \(\ell-1\).  This common
prefix would lie in both \(\Tball(u)\) and \(\Tball(v)\), again contradicting
\(\Tball(u)\cap\Tball(v)=\emptyset\).  Hence no closest witness differs in the
last coordinate, so \textup{(E2)} holds.

Thus the same-length boundary test holds if and only if the truncated balls are
disjoint.
\end{proof}
\fi

\begin{definition}[Two-stage unequal-length test]
\label{def:twostage}
Fix a finite block set \(P\subseteq\Sigma_q^*\).  Let \(x,y\in P\) be blocks of
lengths
\[
    \ell_x=\lvert x\rvert,
    \qquad
    \ell_y=\lvert y\rvert,
    \qquad
    \ell_x<\ell_y.
\]
Put
\[
    d=\ell_y-\ell_x,
    \qquad
    y_0=\prefx_{\ell_x}(y).
\]
The ordered pair \((x,y)\) passes the two-stage unequal-length test if:
\begin{enumerate}[label=\textup{(U\arabic*)}]
    \item \(\Dmin(x,y_0)\ne0\);
    \item if \(\Dmin(x,y_0)=1\), then for every legal prefix
    \(\rho\in\Pref_d(P^*)\), the same-length pair \((x\rho,y)\) passes the
    same-length boundary test.
\end{enumerate}
If \(\Dmin(x,y_0)\ge2\), condition \textup{(U2)} is vacuous.
\end{definition}

The second stage is needed only in the borderline case \(\Dmin=1\).  If the
shorter block and the corresponding prefix of the longer block are already
separated by distance at least two, then a continuation can affect at most one
boundary coordinate and cannot repair both discrepancies.

\begin{lemma}[Cancellation through a common prefix]
\label{lem:cancel}
Let \(a\in\Sigma_q^\ell\), and let \(s,t\in\Sigma_q^n\).  If \(as\) and \(at\)
are confusable in \(\LPC(1)\), then \(s\) and \(t\) are confusable in
\(\LPC(1)\).
\end{lemma}

\ifshowproofs
\begin{proof}
It suffices to prove cancellation of one common initial symbol and then iterate.
Write the common symbol as \(\alpha\), so the two words are \(\alpha s\) and
\(\alpha t\).  Let
\[
    z\in\B(\alpha s)\cap\B(\alpha t).
\]
Choose realizations of \(z\) from both words by pairwise disjoint adjacent
transpositions, omitting ineffective swaps of equal symbols.

We distinguish according to whether the first boundary is used.  If neither
realization uses it, then \(z\) begins with \(\alpha\), and after deleting this
first output symbol we obtain a word in \(\B(s)\cap\B(t)\).

If both realizations use the first boundary, write
\[
    s=\beta s_1,
    \qquad
    t=\gamma t_1.
\]
The first output symbol is \(\beta\) in the first realization and \(\gamma\) in
the second, so \(\beta=\gamma\).  Since the first boundary is used, the second
boundary is not used in either realization.  Thus the common symbol \(\alpha\)
appears as the second output symbol in both realizations.  Deleting this second
output symbol from \(z\) leaves a word obtainable from both
\[
    s=\beta s_1
    \qquad\text{and}\qquad
    t=\beta t_1
\]
by the corresponding shifted sets of disjoint adjacent transpositions.  Hence
\(\B(s)\cap\B(t)\ne\emptyset\).

Finally, exactly one realization cannot use the first boundary.  For example, if
the first boundary is used for \(\alpha s\) but not for \(\alpha t\), then the
first output symbol in the second realization is \(\alpha\), so the first symbol
of \(s\) must also be \(\alpha\).  The swap at the first boundary of
\(\alpha s\) would then be ineffective, contrary to our choice of realization.
The other case is symmetric.

Thus one common initial symbol can be cancelled.  Repeating this \(\ell\) times
cancels the whole common prefix \(a\), and \(s\) and \(t\) are confusable.
\end{proof}
\fi

\begin{theorem}[Two-stage local criterion]
\label{thm:criterion}
Let \(P\subseteq\Sigma_q^*\) be finite.  Assume:
\begin{enumerate}[label=\textup{(C\arabic*)}]
    \item every two distinct blocks \(u,v\in P\) with
    \(\lvert u\rvert=\lvert v\rvert\) pass the same-length boundary test;
    \item every ordered pair \((x,y)\in P^2\) with
    \(\lvert x\rvert<\lvert y\rvert\) passes the two-stage unequal-length test.
\end{enumerate}
Then, for every \(n\ge0\), the finite-length block-concatenation code \(C_n(P)\)
is correcting for \(\LPC(1)\).
\end{theorem}

\ifshowproofs
\begin{proof}
We prove the claim by strong induction on the code length \(n\).  For \(n=0\),
the code \(C_0(P)\) contains only the empty word.  Assume now that, for every
\(n'<n\), the code \(C_{n'}(P)\) is correcting.  Let \(c,c'\in C_n(P)\) be
distinct.  Choose decompositions
\[
    c=as,
    \qquad
    c'=bt,
\]
where \(a,b\in P\) are the first blocks and \(s,t\in P^*\) are the remaining
suffixes.

If \(a=b\), then \(s\ne t\) and
\[
    \lvert s\rvert=\lvert t\rvert=n-\lvert a\rvert<n.
\]
If \(c\) and \(c'\) were confusable, Lemma~\ref{lem:cancel} would imply that
\(s\) and \(t\) are confusable, contradicting the induction hypothesis.  Hence
\(c\) and \(c'\) are not confusable in this case.

It remains to consider \(a\ne b\).  First suppose that
\[
    \lvert a\rvert=\lvert b\rvert=\ell .
\]
If \(c\) and \(c'\) were confusable, then their common received word would have
a prefix of length \(\ell-1\) belonging to both \(\Tball(a)\) and \(\Tball(b)\).
Indeed, a transposition crossing the right boundary of the first block can
affect coordinate \(\ell\), but it cannot affect the first \(\ell-1\)
coordinates.  Thus
\[
    \Tball(a)\cap\Tball(b)\ne\emptyset,
\]
contradicting \textup{(C1)} and Lemma~\ref{lem:trunc}.

Finally suppose, without loss of generality, that
\[
    \ell_a=\lvert a\rvert
    <
    \ell_b=\lvert b\rvert .
\]
Put
\[
    b_0=\prefx_{\ell_a}(b).
\]
If \(c\) and \(c'\) were confusable, then their common received word would have
a prefix of length \(\ell_a-1\) belonging to both \(\Tball(a)\) and
\(\Tball(b_0)\).  Hence, by Lemma~\ref{lem:trunc},
\[
    \Dmin(a,b_0)\le1.
\]
Condition \textup{(U1)} excludes \(\Dmin(a,b_0)=0\), so the only possible case is
\[
    \Dmin(a,b_0)=1.
\]

Let
\[
    d=\ell_b-\ell_a .
\]
Since \(\lvert c\rvert=\lvert c'\rvert=n\), the suffix \(s\) has length at least
\(d\).  Define
\[
    \rho=\prefx_d(s).
\]
Then \(\rho\in\Pref_d(P^*)\), and the two initial windows of length \(\ell_b\)
are
\[
    a\rho
    \qquad\text{and}\qquad
    b.
\]
If \(c\) and \(c'\) were confusable, the common received word would have a prefix
of length \(\ell_b-1\) belonging to both \(\Tball(a\rho)\) and \(\Tball(b)\).
This contradicts \textup{(U2)} and Lemma~\ref{lem:trunc}.  Thus no two distinct
codewords in \(C_n(P)\) are confusable.  By induction, this holds for every
\(n\ge0\).
\end{proof}
\fi

\begin{example}[A valid block set]
\label{ex:valid}
The block set \(P=\{00,11\}\) passes the criterion.  The only same-length pair is
\((00,11)\), and \(\Dmin(00,11)=2\), so the same-length test holds.  There are no
unequal-length pairs.  Hence \(C_n(P)\) is correcting for every \(n\).  This is
the trivial rate-\(1/2\) block construction.
\end{example}

\begin{example}[A failing block set]
\label{ex:failing}
The set \(P_0=\{00,011,111\}\) fails the criterion and indeed does not define a
correcting code family.  In length \(8\), the two codewords
\[
    c=00\,|\,111\,|\,011=00111011,
    \qquad
    c'=011\,|\,00\,|\,111=01100111
\]
belong to \(C_8(P_0)\), and both can reach \(01010111\).  The criterion catches
the failure as follows: appending the legal prefix \(1\) to the shorter block
\(00\) gives \(001\) versus \(011\), for which a closest witness differs in the
last coordinate.
\end{example}

\begin{lemma}[Prefix-freeness follows]
\label{lem:prefixfree}
If \textup{(C2)} holds, then \(P\) is prefix-free.
\end{lemma}

\ifshowproofs
\begin{proof}
If \(x\) is a proper prefix of \(y\), then
\[
    \prefx_{\lvert x\rvert}(y)=x,
\]
so
\[
    \Dmin(x,\prefx_{\lvert x\rvert}(y))=0.
\]
This contradicts condition \textup{(U1)} in the unequal-length test.
\end{proof}
\fi

\begin{remark}[Sufficiency only]
The criterion is a sufficient condition for the finite-length
block-concatenation codes \(C_n(P)\).  It is not a converse for arbitrary
zero-error codes, nor even for all valid block-concatenation code families.  For
a fixed block set \(P\), the exact confusability question is decided by the
verifier of Section~\ref{sec:verifier}.
\end{remark}

\section{Certified binary block constructions}
\label{sec:construction}

We now present the binary block constructions obtained in this work.  For a
finite block set \(P\), the corresponding length-\(n\) code is \(C_n(P)\).  By
Theorem~\ref{thm:criterion}, it is enough to verify the two-stage local
criterion for the finite set \(P\); once this is done, \(C_n(P)\) is correcting
for \(\LPC(1)\) for every \(n\ge0\).

The rate reported for a block set \(P\) is the block/profile rate \(R(P)\)
defined in Section~\ref{sec:prelim}.  Since every block set certified by the
criterion is prefix-free by Lemma~\ref{lem:prefixfree}, this profile rate is the
actual asymptotic rate of the finite-length code family \(\{C_n(P)\}_{n\ge0}\).

The binary upper bound recalled in Section~\ref{sec:prelim} is \(2/3\).
Langberg et al.~\cite{LangbergSchwartzYaakobi2017} constructed binary codes of
rate \(0.609\).  Chee et al.~\cite{CheeKiahLingNguyenVuZhang2016} later improved
this to a string-concatenation construction of rate approximately \(0.642805\).
The best construction below has certified rate \(0.653618\), leaving an
absolute gap of only \(0.013049\) from \(2/3\).  Equivalently, it achieves about
\(98.04\%\) of the known upper bound and closes about \(45\%\) of the gap between
the previous \(0.642805\) construction and \(2/3\).

We first recall the guarded recursive baseline from
\cite{CheeKiahLingNguyenVuZhang2016}.  Let
\[
    P_1=\{000,\ 111,\ 1000,\ 0111,\ 001111,\ 110000\}.
\]
Its length profile is \(p_3=2\), \(p_4=2\), and \(p_6=2\).  Hence the growth
parameter \(\lambda_1\) satisfies
\[
    2\lambda_1^{-3}+2\lambda_1^{-4}+2\lambda_1^{-6}=1,
\]
and
\[
    R(P_1)=\log_2\lambda_1\approx0.642805.
\]
This construction can be viewed as a guarded version of the perfect
length-three code \(\{000,010,101,111\}\), where the unstable short patterns
\(010\) and \(101\) are replaced by longer guarded blocks.  Although the
correctness of \(P_1\) was proved in \cite{CheeKiahLingNguyenVuZhang2016} by a
recursive argument, \(P_1\) also passes our two-stage local criterion, as shown
in Table~\ref{tab:verif} below.

Our first improvement, \(P_2\), follows the same design philosophy.  Just as
\(P_1\) is obtained by guarding the unstable words of the length-three perfect
code, \(P_2\) may be viewed as a further guarded refinement of the \(P_1\)
construction: dangerous adjacencies are replaced by longer bridge blocks while
preserving a finite block-concatenation structure.

The block set \(P_2\) consists of \(30\) blocks, grouped by length as follows.

\begingroup
\small
\setlength{\tabcolsep}{5pt}
\renewcommand{\arraystretch}{1.12}
\begin{center}
\begin{tabular}{@{}c p{0.78\textwidth}@{}}
\toprule
Length & Blocks \\
\midrule
6 &
\(\{000000,\ 000111,\ 001111,\ 110000,\ 111000,\ 111111\}\) \\[1mm]

7 &
\(\{0000111,\ 0001000,\ 0111000,\ 0111111,\ 1000000,\ 1000111,\ 1110111,\ 1111000\}\) \\[1mm]

8 &
\(\{01111000,\ 10000111\}\) \\[1mm]

9 &
\(\{000001111,\ 000110000,\ 001110000,\ 011101111,\ 011111000,\ 100000111,\)
\newline
\(\phantom{\{}100010000,\ 110001111,\ 111001111,\ 111110000\}\) \\[1mm]

11 &
\(\{00110000111,\ 01111001111,\ 10000110000,\ 11001111000\}\) \\
\bottomrule
\end{tabular}
\end{center}
\endgroup

Thus the length profile of \(P_2\) is
\[
    p_6=6,\qquad p_7=8,\qquad p_8=2,\qquad p_9=10,\qquad p_{11}=4.
\]
The growth parameter \(\lambda_2\) is determined by
\[
    6\lambda_2^{-6}
    +8\lambda_2^{-7}
    +2\lambda_2^{-8}
    +10\lambda_2^{-9}
    +4\lambda_2^{-11}
    =
    1,
\]
and
\[
    R(P_2)=\log_2\lambda_2\approx0.649872.
\]

The same search-and-certify method gives two larger block sets with still higher
rates.  The first, denoted \(P_3\), has block lengths \(9,\dots,15\) and profile
\[
    (p_9,p_{10},p_{11},p_{12},p_{13},p_{14},p_{15})
    =
    (18,26,18,40,22,14,18),
\]
with rate
\[
    R(P_3)\approx0.652018.
\]
The best block set found, denoted \(P^\star\), has block lengths \(10,\dots,19\)
and profile
\[
\begin{aligned}
    &(p_{10},p_{11},p_{12},p_{13},p_{14},p_{15},p_{16},p_{17},p_{18},p_{19})\\
    &\qquad=(14,6,78,58,56,92,60,72,58,80),
\end{aligned}
\]
with rate
\[
    R(P^\star)\approx0.653618.
\]
Since \(P^\star\) is the best certified construction and the explicit block
lists are long, we record only the profile and rate of \(P_3\).  The explicit
block list for \(P^\star\) is given in Appendix~\ref{app:blocks}.

Table~\ref{tab:profiles} summarizes the profiles and rates of the four block
sets.  The column \(|P|\) gives the total number of blocks, and the profile
column records the number \(p_\ell\) of blocks of each relevant length
\(\ell\).  The rate is computed from
\[
    \sum_{\ell}p_\ell\lambda^{-\ell}=1,
    \qquad
    R(P)=\log_2\lambda,
\]
and the last column gives the absolute gap between \(R(P)\) and the upper bound
\(2/3\).

\begin{table}[!t]
    \centering
    \small
    \begin{tabular}{@{}l c l c c@{}}
        \toprule
        Block set & \(|P|\) & profile \((p_\ell)\) & rate & gap to \(2/3\) \\
        \midrule
        \(P_1\) & 6
            & \(p_3=2,\ p_4=2,\ p_6=2\)
            & \(0.642805\)
            & \(0.023862\) \\
        \(P_2\) & 30
            & \(p_6=6,\ p_7=8,\ p_8=2,\ p_9=10,\ p_{11}=4\)
            & \(0.649872\)
            & \(0.016795\) \\
        \(P_3\) & 156
            & \((18,26,18,40,22,14,18)\), \(9\le\ell\le15\)
            & \(0.652018\)
            & \(0.014649\) \\
        \(P^\star\) & 574
            & \((14,6,78,58,56,92,60,72,58,80)\), \(10\le\ell\le19\)
            & \(0.653618\)
            & \(0.013049\) \\
        \bottomrule
    \end{tabular}
    \caption{Block profiles, certified rates, and gaps to the upper bound \(2/3\).}
    \label{tab:profiles}
\end{table}

For each block set in Table~\ref{tab:profiles}, validity follows by applying
the two-stage local criterion.  Table~\ref{tab:verif} records the corresponding
finite certificate.  The column ``same-length pairs'' counts unordered pairs of
distinct blocks of the same length that must pass the same-length boundary test.
For example, \(P_1\) has two blocks of each of the lengths \(3,4,6\), hence
\(1+1+1=3\) same-length pairs.  The column ``unequal pairs'' counts ordered
pairs \((x,y)\) with \(|x|<|y|\).  For \(P_1\), this gives
\(2\cdot2+2\cdot2+2\cdot2=12\) pairs.  The column ``first-stage safe'' counts
unequal pairs accepted immediately because
\[
    \Dmin\bigl(x,\prefx_{\lvert x\rvert}(y)\bigr)\ge2.
\]
The remaining pairs are sent to the second stage, where all legal prefixes
\(\rho\in\Pref_d(P^*)\) are checked.  The last column records the total number
of same-length boundary tests performed in this second stage.

\begin{table}[!t]
    \centering
    \small
    \begin{tabular}{@{}l r r r r r@{}}
        \toprule
        Block set
            & same-length pairs
            & unequal pairs
            & first-stage safe
            & sent to stage 2
            & stage-2 checks \\
        \midrule
        \(P_1\)       & 3     & 12     & 4      & 8    & 32 \\
        \(P_2\)       & 95    & 340    & 266    & 74   & 364 \\
        \(P_3\)       & 1886  & 10204  & 9588   & 616  & 3624 \\
        \(P^\star\)   & 19627 & 144824 & 141776 & 3048 & 37944 \\
        \bottomrule
    \end{tabular}
    \caption{Verification of the two-stage criterion for the block sets in Table~\ref{tab:profiles}.}
    \label{tab:verif}
\end{table}

Therefore, by Theorem~\ref{thm:criterion}, for every length \(n\), the codes
\[
    C_n(P_1),\qquad C_n(P_2),\qquad C_n(P_3),\qquad C_n(P^\star)
\]
are correcting for \(\LPC(1)\).  In particular, \(P^\star\) gives a certified
binary block-concatenation family of rate \(0.653618\), improving the previous
string-concatenation rate \(0.642805\) and coming within about \(2\%\) of the
upper bound \(2/3\).

\FloatBarrier

\section{Properties and limitations of the local criterion}
\label{sec:criterion-properties}

The constructions of Section~\ref{sec:construction} show that the two-stage
criterion can certify concrete binary block-concatenation codes of rate
\(0.653618\).  We now ask a more structural question: is this criterion itself
a possible obstacle to reaching the true zero-error capacity?

The answer is no.  Although the criterion is only sufficient and may reject some
valid block sets, we prove that it is \emph{rate-complete}: for every alphabet
size \(q\), the supremum of rates of block sets certified by the criterion is
exactly the zero-error capacity \(C_0^{(q)}\).  Thus any asymptotic rate loss in
our explicit finite constructions comes from the particular block families we
found, not from the certification method.

The proof has two steps.  First, we analyze the special case in which all blocks
have one fixed length.  In this case the second stage of the criterion is
irrelevant, and certification is exactly the problem of finding a large family
of blocks with pairwise disjoint truncated balls.  We show that this problem is
equivalent to ordinary zero-error coding one coordinate shorter.  Second, by
padding an optimal length-\(\ell\) code with one repeated final symbol and then
concatenating the padded words, we obtain certified block constructions whose
rates approach \(C_0^{(q)}\).

After proving this rate-completeness result, we record two finite verification
constraints that help interpret the computational tables in
Section~\ref{sec:construction}: a simple composition filter and a necessary
per-length packing condition.

\subsection{Single-length certification and rate-completeness}

We begin with block sets whose blocks all have the same length.  This is the
cleanest setting because unequal-length interactions do not occur, so condition
\textup{(C2)} of Theorem~\ref{thm:criterion} is vacuous.  The criterion then
reduces to the same-length boundary test, or equivalently to disjointness of
truncated balls.

Call a set \(P\subseteq\Sigma_q^\ell\) \emph{\(\Tball\)-disjoint} if the sets
\(\Tball(u)\), \(u\in P\), are pairwise disjoint.  By Lemma~\ref{lem:trunc}, this
is precisely the same-length certification condition for a block set all of
whose blocks have length \(\ell\).  Let \(M_{\Tball,q}(\ell)\) denote the maximum
size of a \(\Tball\)-disjoint subset of \(\Sigma_q^\ell\).

\begin{theorem}[Single-length certification equals shorter zero-error coding]
\label{thm:MT}
For every \(q\ge2\) and every \(\ell\ge1\),
\[
    M_{\Tball,q}(\ell)=A_q(\ell-1;1).
\]
\end{theorem}

\ifshowproofs
\begin{proof}
For \(\ell=1\), every length-one word \(u\) has
\(\Tball(u)=\{\epsilon\}\).  Hence a \(\Tball\)-disjoint family has size at most
one, and
\[
    M_{\Tball,q}(1)=1=A_q(0;1).
\]
Assume now that \(\ell\ge2\).

For the lower bound, let \(C\subseteq\Sigma_q^{\ell-1}\) be an optimal correcting
code.  Append to each \(c=c_1\cdots c_{\ell-1}\in C\) a copy of its last symbol,
and put
\[
    P=\{c\,c_{\ell-1}:c\in C\}\subseteq\Sigma_q^\ell.
\]
The appended boundary is not a transition, so
\[
    \Tball(c\,c_{\ell-1})=\B(c).
\]
Since the balls \(\B(c)\), \(c\in C\), are pairwise disjoint, the family \(P\) is
\(\Tball\)-disjoint and has size \(A_q(\ell-1;1)\).

For the upper bound, let \(P\subseteq\Sigma_q^\ell\) be
\(\Tball\)-disjoint.  The map
\[
    u\longmapsto \bar u=\prefx_{\ell-1}(u)
\]
is injective, because \(\bar u\in\Tball(u)\), and equal prefixes would give an
intersection of truncated balls.  Moreover,
\[
    \B(\bar u)\subseteq\Tball(u).
\]
Hence the balls \(\B(\bar u)\), \(u\in P\), are pairwise disjoint, and
\[
    \{\bar u:u\in P\}\subseteq\Sigma_q^{\ell-1}
\]
is a correcting code of size \(\lvert P\rvert\).  Therefore
\[
    \lvert P\rvert\le A_q(\ell-1;1).
\]
\end{proof}
\fi

\begin{example}[Why \(010\) and \(101\) are unstable]
\label{ex:perfect3}
The perfect length-three code \(\{000,010,101,111\}\) has size \(A_2(3;1)=4\),
but it is not \(\Tball\)-disjoint:
\[
    \Tball(010)=\{01,10,00\},
    \qquad
    \Tball(101)=\{10,01,11\}.
\]
Thus
\[
    \Tball(010)\cap\Tball(101)\ne\emptyset,
\]
and also
\[
    00\in\Tball(010)\cap\Tball(000).
\]
By Theorem~\ref{thm:MT},
\[
    M_{\Tball,2}(3)=A_2(2;1)=3.
\]
Hence the two non-constant words \(010\) and \(101\) cannot both be used as
unguarded length-three blocks.  The guarded baseline \(P_1\) resolves exactly
this obstruction.
\end{example}

Theorem~\ref{thm:MT} turns any optimal zero-error code of length \(\ell\) into a
certified single-length block set of block length \(\ell+1\), by appending one
copy of the last symbol.  Concatenating \(k\) such certified blocks gives the
following buffered supermultiplicativity inequality.

\begin{corollary}[Buffered supermultiplicativity]
\label{cor:supermult}
For every \(q\ge2\) and all \(\ell,k\ge1\),
\[
    A_q(k(\ell+1);1)\ge A_q(\ell;1)^k.
\]
Consequently,
\[
    C_0^{(q)}
    =
    \sup_{\ell\ge1}
    \frac{1}{\ell+1}\log_q A_q(\ell;1).
\]
\end{corollary}

\ifshowproofs
\begin{proof}
Let \(C\subseteq\Sigma_q^\ell\) be an optimal correcting code of length
\(\ell\).  As in the proof of Theorem~\ref{thm:MT}, append to each
\(c=c_1\cdots c_\ell\in C\) a copy of its last symbol and obtain
\[
    P=\{c\,c_\ell:c\in C\}\subseteq\Sigma_q^{\ell+1}.
\]
This block set is \(\Tball\)-disjoint.  Since all blocks in \(P\) have the same
length, condition \textup{(C2)} is vacuous, and condition \textup{(C1)} follows
from Lemma~\ref{lem:trunc}.  Thus Theorem~\ref{thm:criterion} certifies the
finite-length codes \(C_n(P)\).  In particular, the words formed by
concatenating exactly \(k\) blocks from \(P\) form a correcting code of length
\(k(\ell+1)\) and size
\[
    \lvert P\rvert^k=A_q(\ell;1)^k.
\]
This proves the supermultiplicative inequality.

For the capacity identity, write
\[
    a_n=\log_q A_q(n;1).
\]
The inequality above gives
\[
    \frac{1}{k(\ell+1)}a_{k(\ell+1)}
    \ge
    \frac{1}{\ell+1}a_\ell
\]
for every \(k\).  Taking the limsup along the subsequence \(k(\ell+1)\) yields
\[
    C_0^{(q)}
    \ge
    \frac{1}{\ell+1}a_\ell
\]
for each \(\ell\), and hence
\[
    C_0^{(q)}
    \ge
    \sup_{\ell\ge1}\frac{1}{\ell+1}a_\ell .
\]
Conversely, for every \(n\ge1\),
\[
    \frac{1}{n}a_n
    =
    \frac{n+1}{n}\cdot\frac{1}{n+1}a_n
    \le
    \frac{n+1}{n}
    \sup_{\ell\ge1}\frac{1}{\ell+1}a_\ell .
\]
Letting \(n\to\infty\) gives the reverse inequality.
\end{proof}
\fi

Let \(\rho_{\mathrm{crit}}^{(q)}\) be the supremum of the block rates \(R(P)\),
where \(P\) ranges over all finite \(q\)-ary block sets certified by the
two-stage local criterion.

\begin{theorem}[Rate-completeness of the criterion]
\label{thm:complete}
For every \(q\ge2\), the two-stage local criterion is rate-complete:
\[
    \rho_{\mathrm{crit}}^{(q)}=C_0^{(q)}.
\]
In particular, in the binary case,
\[
    \rho_{\mathrm{crit}}^{(2)}=C_0.
\]
\end{theorem}

\ifshowproofs
\begin{proof}
First, suppose that a finite block set \(P\subseteq\Sigma_q^*\) is certified by
the two-stage criterion.  Then Theorem~\ref{thm:criterion} implies that
\(C_n(P)\) is correcting for every \(n\).  Hence
\[
    \lvert C_n(P)\rvert\le A_q(n;1)
\]
for every admissible length \(n\).  Since certified block sets are prefix-free by
Lemma~\ref{lem:prefixfree}, the profile rate \(R(P)\) is the actual asymptotic
rate of the family \(C_n(P)\).  Therefore
\[
    R(P)\le C_0^{(q)}.
\]
Taking the supremum over all certified \(P\) gives
\[
    \rho_{\mathrm{crit}}^{(q)}\le C_0^{(q)}.
\]

For the reverse inequality, let \(C\subseteq\Sigma_q^\ell\) be an optimal
correcting code of length \(\ell\), and form the padded block set
\[
    P=\{c\,c_\ell:c=c_1\cdots c_\ell\in C\}
    \subseteq\Sigma_q^{\ell+1}.
\]
By Theorem~\ref{thm:MT}, this set is \(\Tball\)-disjoint.  Since all blocks have
the same length, the two-stage criterion certifies \(P\).  Its block rate is
\[
    R(P)=\frac{1}{\ell+1}\log_q A_q(\ell;1).
\]
Taking the supremum over \(\ell\) and applying Corollary~\ref{cor:supermult}
gives
\[
    \rho_{\mathrm{crit}}^{(q)}\ge C_0^{(q)}.
\]
\end{proof}
\fi

Thus the criterion itself is not the asymptotic bottleneck for any alphabet
size.  Improving certified rates is, in principle, exactly as hard as improving
the zero-error capacity \(C_0^{(q)}\).

\subsection{Finite verification filters and packing}

The rate-completeness theorem is asymptotic.  We now record two finite
constraints that help interpret the verification data in
Table~\ref{tab:verif}.  The first is an elementary filter used in the first
stage of the criterion.  For \(w\in\Sigma_q^\ell\), write
\[
    \operatorname{comp}(w)
    =
    (n_0(w),n_1(w),\dots,n_{q-1}(w)),
\]
where \(n_a(w)\) is the number of occurrences of \(a\) in \(w\).  Since
\(\LPC(1)\) only permutes coordinates, every word in \(\B(w)\) has the same
composition as \(w\).  Therefore, for equal-length words \(u,v\in\Sigma_q^\ell\),
\[
    \Dmin(u,v)
    \ge
    \frac12
    \left\lVert
        \operatorname{comp}(u)-\operatorname{comp}(v)
    \right\rVert_1 .
\]
In the binary case this becomes
\[
    \Dmin(u,v)\ge\lvert\wt(u)-\wt(v)\rvert.
\]
Thus any binary pair with weight difference at least two is automatically
first-stage safe.  This explains why most unequal-length pairs in
Table~\ref{tab:verif} are accepted before the second stage is invoked.

We next give a necessary packing condition for blocks of a fixed length.  The
ball-size formula below is a restatement, in the present transition-run notation,
of the antirun-profile formula of Langberg et
al.~\cite[Theorem~10]{LangbergSchwartzYaakobi2017}.  We include the short
argument only to make the indexing used below clear.

For a word \(w=w_1\cdots w_\ell\), let
\[
    \Tr(w)=\{i\in\{1,\dots,\ell-1\}:w_i\ne w_{i+1}\}
\]
be its set of transition boundaries.

\begin{lemma}[Ball size]
\label{lem:ballsize}
Let \(w\in\Sigma_q^\ell\), and let \(t_1,\dots,t_s\) be the lengths of the
maximal runs of consecutive transition boundaries in \(\Tr(w)\).  Then
\[
    \lvert\B(w)\rvert=\prod_{j=1}^{s}F_{t_j+2},
\]
where \(F_1=F_2=1\) and \(F_t=F_{t-1}+F_{t-2}\) for \(t\ge3\).
\end{lemma}

\ifshowproofs
\begin{proof}
Each \(z\in\B(w)\) is obtained by swapping a set \(S\) of pairwise non-adjacent
adjacent boundaries.  A swap at a non-transition boundary has no effect, so we
may take
\[
    S\subseteq\Tr(w)
\]
with no two adjacent boundaries.  Conversely, every such set \(S\) gives a word
in \(\B(w)\).

Distinct such sets \(S\) give distinct received words.  Indeed, let \(i\) be the
smallest boundary in the symmetric difference of two such sets, and suppose
\(i\in S\) for one set but not the other.  Since \(i\in\Tr(w)\), the symbols
\(w_i\) and \(w_{i+1}\) are different.  Moreover, by minimality of \(i\) and
non-adjacency of the chosen boundary sets, the two resulting words differ at
coordinate \(i\).  Thus the number of distinct received words equals the number
of subsets of \(\Tr(w)\) with no adjacent elements.

Inside a maximal run of \(t\) consecutive transition boundaries, this is the
number of independent sets in a path on \(t\) vertices, namely \(F_{t+2}\).
Different runs are separated by gaps, so the counts multiply.
\end{proof}
\fi

For a finite block set \(P\), write
\[
    P_\ell=\{u\in P:\lvert u\rvert=\ell\}.
\]

\begin{proposition}[Per-length packing]
\label{prop:packing}
If \(P\subseteq\Sigma_q^*\) satisfies the two-stage criterion, then for every
\(\ell\ge1\),
\[
    \sum_{u\in P_\ell}\lvert\B(u)\rvert\le q^{\,\ell-1}.
\]
Equivalently, if \(t_1(u),\dots,t_{s(u)}(u)\) are the transition-run lengths of
\(u\), then
\[
    \sum_{u\in P_\ell}\prod_{j=1}^{s(u)}F_{t_j(u)+2}\le q^{\,\ell-1}.
\]
\end{proposition}

\ifshowproofs
\begin{proof}
If \(P\) satisfies the criterion, then every pair of distinct blocks in
\(P_\ell\) passes the same-length boundary test.  By Lemma~\ref{lem:trunc}, the
truncated balls \(\Tball(u)\), \(u\in P_\ell\), are pairwise disjoint subsets of
\(\Sigma_q^{\ell-1}\).  Moreover,
\[
    \lvert\Tball(u)\rvert=\lvert\B(u)\rvert.
\]
Therefore
\[
    \sum_{u\in P_\ell}\lvert\B(u)\rvert
    =
    \sum_{u\in P_\ell}\lvert\Tball(u)\rvert
    \le
    \lvert\Sigma_q^{\ell-1}\rvert
    =
    q^{\ell-1}.
\]
The Fibonacci product form follows from Lemma~\ref{lem:ballsize}.
\end{proof}
\fi

The packing condition is only necessary.  It does not capture the full
confusability behavior across different block lengths, which is precisely why
the second stage of the criterion is needed.

\section{The exact product-automaton verifier}
\label{sec:verifier}

The two-stage criterion of Section~\ref{sec:criterion} gives a short sufficient
certificate for the finite-length codes \(C_n(P)\), but it is not an exact test.
We now describe an exact finite-state verifier for the binary case.  Given a
finite prefix-free block set \(P\subseteq\Sigma^*\setminus\{\epsilon\}\), where
\(\Sigma=\{0,1\}\), the verifier decides whether, for every \(n\ge0\), the code
\(C_n(P)\) is correcting for \(\LPC(1)\).

Equivalently, the verifier decides whether there exist a length \(n\), two
distinct codewords \(c_1,c_2\in C_n(P)\), and a received word \(z\in\Sigma^n\)
such that
\[
    z\in\B(c_1)\cap\B(c_2).
\]
Unlike the local criterion, the verifier tests this confusability condition
directly.  A \texttt{SUCCESS} verdict therefore certifies all lengths
simultaneously.

The idea is as follows.  A received word \(z\) is a possible channel output of a
codeword \(c\in C_n(P)\) precisely when \(z\in\B(c)\) and \(c\in P^*\).  We build
a finite-state parser that scans \(z\) from left to right and
nondeterministically reconstructs a possible transmitted word \(c\) satisfying
both conditions.  The family \(\{C_n(P)\}_{n\ge0}\) fails to be correcting exactly
when the same received word \(z\) admits two different such reconstructions.
Thus the verifier runs two copies of the parser on the same input \(z\) and
searches for a reachable product state in which the two reconstructed codewords
are distinct.

\subsection{Reconstructing transmitted words}

The reconstruction uses two elementary components.  Domino tilings describe
radius-one error balls, and a trie automaton recognizes the language \(P^*\).

\begin{definition}[Domino tiling]
\label{def:domino-tiling}
A \emph{domino tiling} \(\tau\) of \([n]\) is a partition of \([n]\) into
singletons \(\{i\}\) and adjacent dominoes \(\{i,i+1\}\).  For a word
\(c=c_1\cdots c_n\in\Sigma^n\), define \(E_\tau(c)\) by fixing the symbols on
singleton parts and swapping the two symbols on every domino part.
\end{definition}

\begin{lemma}[Balls and tilings]
\label{lem:tilings}
For every \(c\in\Sigma^n\),
\[
    \B(c)=\{E_\tau(c):\tau\text{ is a domino tiling of }[n]\}.
\]
Moreover, \(E_\tau\) is an involution.  Hence \(z\in\B(c)\) if and only if
\(c=E_\tau(z)\) for some domino tiling \(\tau\).
\end{lemma}

\ifshowproofs
\begin{proof}
Every domino tiling describes a set of pairwise disjoint adjacent transpositions,
so \(E_\tau(c)\in\B(c)\).  Conversely, let \(\pi\in S_n\) satisfy
\[
    \max_i\lvert\pi(i)-i\rvert\le1.
\]
If \(\pi(i)=i+1\), then necessarily \(\pi(i+1)=i\); otherwise injectivity of
\(\pi\) would force an impossible chain of shifts.  Hence every non-fixed point
belongs to an adjacent transposition, and the remaining points are fixed.
Applying the same swaps twice gives the identity, so \(E_\tau\) is an involution.
\end{proof}
\fi

We next define the automaton that recognizes \(P^*\).  Let
\[
    \mathcal A_P=(V,\delta,r_0)
\]
be the trie automaton of \(P\).  The root state is \(r_0\), corresponding to the
empty prefix \(\epsilon\).  The set \(V\) consists of all proper prefixes of
blocks in \(P\), including \(\epsilon\).  We identify each state \(v\in V\) with
the prefix word it represents.

The transition function
\[
    \delta:V\times\Sigma\to V
\]
is partial.  If the prefix represented by \(v\), followed by \(b\in\Sigma\), is a
proper prefix of a block in \(P\), then \(\delta(v,b)\) is the state representing
that new prefix.  If appending \(b\) completes a block of \(P\), then
\[
    \delta(v,b)=r_0.
\]
Otherwise \(\delta(v,b)\) is undefined.  We write \(\widehat\delta\) for the
natural extension of \(\delta\) to words.

\begin{lemma}[Trie recognizer]
\label{lem:trie}
If \(P\) is prefix-free, then
\[
    w\in P^*
    \qquad\Longleftrightarrow\qquad
    \widehat\delta(r_0,w)=r_0.
\]
\end{lemma}

\ifshowproofs
\begin{proof}
Whenever a block is completed, prefix-freeness guarantees that there is no
ambiguity between stopping the current block and continuing it; the transition
returns to the root.  Therefore every concatenation of blocks returns to the
root after each block and at the end.  Conversely, if reading \(w\) from the root
ends at the root, then the successive returns to the root decompose \(w\) into
blocks from \(P\).
\end{proof}
\fi

The parser reads the received word \(z\) one symbol at a time and feeds a
reconstructed transmitted word into \(\mathcal A_P\).  A parser micro-state is a
pair
\[
    (v,p),
\]
where \(v\in V\) is the current trie state and
\[
    p\in\{\circ,H_0,H_1\}.
\]
The symbol \(\circ\) means that no swap is pending.  For \(h\in\Sigma\), the
state \(H_h\) means that the parser has already read the first received symbol
\(h\) of a domino, has not yet emitted it, and is waiting for the next received
symbol before emitting the swapped pair.

On reading a received symbol \(\beta\), the allowed parser moves are:
\begin{itemize}[leftmargin=*]
    \item \textbf{copy} from \((v,\circ)\): emit \(\beta\) and move to
    \[
        (\delta(v,\beta),\circ),
    \]
    provided \(\delta(v,\beta)\) is defined;

    \item \textbf{swap-start} from \((v,\circ)\): emit nothing and move to
    \[
        (v,H_\beta);
    \]

    \item \textbf{swap-resolve} from \((v,H_h)\): emit \(\beta h\) and move to
    \[
        (\delta(\delta(v,\beta),h),\circ),
    \]
    provided both trie transitions are defined.
\end{itemize}
Thus the parser consumes exactly one received symbol per step, but emits either
\(0\), \(1\), or \(2\) reconstructed symbols.

\begin{lemma}[Parser--ball correspondence]
\label{lem:parser}
Fix \(z\in\Sigma^n\).  The accepting parser runs on \(z\), from
\((r_0,\circ)\) to \((r_0,\circ)\), are in bijection with pairs \((\tau,w)\),
where \(\tau\) is a domino tiling of \([n]\), \(w=E_\tau(z)\), and \(w\in P^*\).
Consequently, there is an accepting parser run on \(z\) feeding the word \(w\) if
and only if \(w\in P^*\) and \(z\in\B(w)\).
\end{lemma}

\ifshowproofs
\begin{proof}
A copy move corresponds to a singleton of the tiling and feeds the current
received symbol unchanged.  A swap-start followed by a swap-resolve corresponds
to a domino: if the two received symbols are \(h,\beta\), then the parser feeds
\(\beta h\), which is exactly the inverse adjacent transposition.  Since
\(E_\tau\) is an involution by Lemma~\ref{lem:tilings}, the word fed into the
trie is \(E_\tau(z)\).  The parser ends at \((r_0,\circ)\) exactly when there is
no unresolved domino and the fed word is accepted by the trie.  By
Lemma~\ref{lem:trie}, this is equivalent to the fed word belonging to \(P^*\).
\end{proof}
\fi

\subsection{The product automaton}

The exact verifier runs two copies of the parser on the same received word
\(z\), in lockstep.  At each step, both parsers consume the same received symbol
\(\beta\in\Sigma\), but they may emit different numbers of reconstructed
symbols.  The product automaton compares the two reconstructed words online,
while storing at most one unmatched emitted symbol.

The synchronization invariant is the following.  After \(k\) lockstep steps,
each parser has consumed \(k\) received symbols.  A parser whose pending
component is \(\circ\) has emitted \(k\) reconstructed symbols, while a parser
whose pending component is \(H_0\) or \(H_1\) has emitted \(k-1\) reconstructed
symbols.  Hence the two emitted prefixes differ in length by at most one.
Therefore one unmatched symbol is enough for synchronization.

We define the product automaton \(\mathcal M_P\).  Its states are tuples
\[
    (v_1,p_1,v_2,p_2,\mathrm{div},\mathrm{lead}),
\]
where
\[
    v_1,v_2\in V,
    \qquad
    p_1,p_2\in\{\circ,H_0,H_1\},
\]
\[
    \mathrm{div}\in\{\bot,\top\},
\]
and
\[
    \mathrm{lead}\in
    \{\circ,L_1(0),L_1(1),L_2(0),L_2(1)\}.
\]
The component \(\mathrm{lead}\) is a one-symbol buffer used to synchronize the two
reconstructed words.  It records whether one parser has emitted one symbol that
has not yet been matched with a symbol emitted by the other parser.  The value
\(\mathrm{lead}=\circ\) means that the two reconstructed prefixes are currently
aligned.  The value \(L_i(a)\), where \(i\in\{1,2\}\) and \(a\in\Sigma\), means
that parser \(i\) is one symbol ahead and that its unmatched emitted symbol is
\(a\).

The flag \(\mathrm{div}\) records whether the two reconstructed words have
already differed at an aligned coordinate: \(\bot\) means that no difference has
yet been detected, whereas \(\top\) means that such a difference has been
detected.  The initial state is
\[
    (r_0,\circ,r_0,\circ,\bot,\circ).
\]

We now define the transitions.  Fix a received symbol \(\beta\in\Sigma\).  Each
parser chooses one legal parser move on input \(\beta\).  Suppose parser \(i\)
moves from \((v_i,p_i)\) to \((v_i',p_i')\) and emits a word
\[
    \eta_i\in\Sigma^{\le2}.
\]
Thus \(\eta_i\) is \(\beta\) for a copy move, \(\epsilon\) for a swap-start move,
and \(\beta h\) for a swap-resolve move from pending state \(H_h\).

The previous lead is attached to the appropriate side.  Define
\(\xi_1,\xi_2\in\Sigma^*\) by
\[
\begin{cases}
\xi_1=\eta_1,\quad \xi_2=\eta_2,
    & \text{if } \mathrm{lead}=\circ,\\[1mm]
\xi_1=a\eta_1,\quad \xi_2=\eta_2,
    & \text{if } \mathrm{lead}=L_1(a),\\[1mm]
\xi_1=\eta_1,\quad \xi_2=a\eta_2,
    & \text{if } \mathrm{lead}=L_2(a).
\end{cases}
\]
Compare \(\xi_1\) and \(\xi_2\) from left to right for as long as both have a
symbol available.  If any compared pair is unequal, set
\[
    \mathrm{div}'=\top;
\]
otherwise keep the previous value, \(\mathrm{div}'=\mathrm{div}\).

After the common aligned prefix has been compared and deleted, the
synchronization invariant implies that at most one symbol remains.  If no symbol
remains, set \(\mathrm{lead}'=\circ\).  If one symbol \(a\) remains on the first
parser side, set \(\mathrm{lead}'=L_1(a)\).  If one symbol \(a\) remains on the
second parser side, set \(\mathrm{lead}'=L_2(a)\).  This gives a transition from
\[
    (v_1,p_1,v_2,p_2,\mathrm{div},\mathrm{lead})
\]
to
\[
    (v_1',p_1',v_2',p_2',\mathrm{div}',\mathrm{lead}').
\]
Only transitions arising from legal parser moves are included.

A product state is accepting if
\[
    v_1=v_2=r_0,
    \qquad
    p_1=p_2=\circ,
    \qquad
    \mathrm{lead}=\circ,
    \qquad
    \mathrm{div}=\top.
\]
The first three conditions say that both parsers have finished clean
reconstructions of words in \(P^*\) of the same length.  The condition
\(\mathrm{div}=\top\) says that the two reconstructed words are distinct.

\begin{theorem}[The product automaton decides confusability]
\label{thm:verifier}
Assume \(P\subseteq\Sigma^*\setminus\{\epsilon\}\) is finite and prefix-free.  An
accepting state of \(\mathcal M_P\) is reachable if and only if there exist a
length \(n\), two distinct codewords \(c_1,c_2\in C_n(P)\), and a word
\(z\in\Sigma^n\) such that
\[
    z\in\B(c_1)\cap\B(c_2).
\]
Consequently, \(C_n(P)\) is correcting for every \(n\ge0\) if and only if no
accepting product state is reachable.
\end{theorem}

\ifshowproofs
\begin{proof}
Suppose first that an accepting product state is reachable after reading some
received word \(z\).  Since both parsers end at \((r_0,\circ)\),
Lemma~\ref{lem:parser} gives two words \(c_1,c_2\in P^*\) such that
\[
    z\in\B(c_1)\cap\B(c_2).
\]
The accepting conditions \(p_1=p_2=\circ\) and \(\mathrm{lead}=\circ\) imply, by
the synchronization invariant, that both reconstructed words have the same
length as \(z\).  Hence
\[
    c_1,c_2\in C_{\lvert z\rvert}(P).
\]
Since \(\mathrm{lead}=\circ\), every aligned coordinate has been compared, and
\(\mathrm{div}=\top\) means that some compared coordinate differed.  Therefore
\(c_1\ne c_2\).

Conversely, suppose that \(c_1,c_2\in C_n(P)\) are distinct and that
\(z\in\B(c_1)\cap\B(c_2)\).  By Lemma~\ref{lem:tilings}, there are domino tilings
\(\tau_1,\tau_2\) with
\[
    c_1=E_{\tau_1}(z),
    \qquad
    c_2=E_{\tau_2}(z).
\]
By Lemma~\ref{lem:parser}, these tilings define two accepting parser runs on
\(z\).  Running the two parser runs in lockstep gives a path in
\(\mathcal M_P\).  The synchronization invariant guarantees that the lead
component never needs to store more than one symbol.  Since \(c_1\ne c_2\), some
aligned coordinate differs, so \(\mathrm{div}\) is eventually set to \(\top\).
Both parser runs end at \((r_0,\circ)\), and the final lead is \(\circ\).
Therefore the terminal product state is accepting.
\end{proof}
\fi

\begin{theorem}[Bound-free all-length certificate]
\label{thm:terminate}
Let \(\mathcal A_P=(V,\delta,r_0)\) be the trie automaton of a finite prefix-free
binary block set \(P\).  The product automaton \(\mathcal M_P\) has at most
\(90\lvert V\rvert^2\) states.  Therefore a breadth-first search over reachable
product states terminates using \(O(\lvert V\rvert^2)\) memory, independently of
any codeword length.  The search returns \texttt{SUCCESS} if and only if
\(C_n(P)\) is correcting for every \(n\ge0\); otherwise it returns a shortest
collision.
\end{theorem}

\ifshowproofs
\begin{proof}
There are \(\lvert V\rvert^2\) choices for the two trie states
\((v_1,v_2)\), \(3^2\) choices for the two pending states \((p_1,p_2)\), two
choices for \(\mathrm{div}\), and five choices for \(\mathrm{lead}\).  Hence the
number of product states is at most
\[
    \lvert V\rvert^2\cdot 3^2\cdot2\cdot5
    =
    90\lvert V\rvert^2.
\]
The transition relation is finite, so breadth-first search from the initial
state
\[
    (r_0,\circ,r_0,\circ,\bot,\circ)
\]
terminates after exploring at most this many states.  The correctness and
shortest-collision statements follow from Theorem~\ref{thm:verifier}; breadth
first search returns a collision of minimum received length when one exists.
\end{proof}
\fi

\subsection{Comparison with the local criterion}

The product automaton and the two-stage local criterion serve different roles.
The criterion is the main mathematical certificate used for the constructions in
Section~\ref{sec:construction}: it gives a short, checkable proof based only on
local block-pair tests.  The verifier is exact for a fixed prefix-free block set:
it tests ball-disjointness itself rather than a sufficient local condition.

Both tools certify all lengths simultaneously.  The criterion does this through
the induction in Theorem~\ref{thm:criterion}; the verifier does it through
finite-state reachability in Theorem~\ref{thm:terminate}.  Neither requires a
length cutoff.

Neither tool imposes an asymptotic rate ceiling below the true zero-error
capacity.  For the criterion this is Theorem~\ref{thm:complete}, which gives
\[
    \rho_{\mathrm{crit}}^{(q)}=C_0^{(q)}
\]
for every alphabet size \(q\).  The same conclusion holds for the binary
verifier, since every criterion-certified binary block set is accepted by the
exact verifier, while any accepted block set gives correcting codes and therefore
cannot have rate exceeding \(C_0\).

The criterion is \(q\)-ary as stated in Section~\ref{sec:criterion}.  The
verifier above is written for the binary case, where a pending swap stores one
bit.  The same product-automaton idea extends directly to a \(q\)-ary alphabet:
the pending state stores one symbol, and the one-symbol lead buffer has
\(1+2q\) possibilities.  Thus the binary state bound \(90\lvert V\rvert^2\)
becomes
\[
    2(q+1)^2(1+2q)\lvert V\rvert^2
\]
states in the direct \(q\)-ary generalization.

\section{Error detection}
\label{sec:detect}

As in classical coding theory, we also study the corresponding error-detection
problem.  Detection is operationally natural in settings with feedback: the
receiver may not be able, or may not need, to recover the transmitted word from
the corrupted output, but it can recognize that an error occurred and request a
retransmission from the sender.

At first sight, detection may appear to be only an easier variant of correction.
For the radius-one limited permutation channel, however, it leads to a genuinely
different problem.  For correction, the error balls around distinct codewords
must be disjoint.  For detection, intersections between error balls are allowed;
one only has to prevent a nontrivial channel output from being another codeword.
Thus the relevant obstruction is the directed condition \(c'\in\B(c)\), rather
than the symmetric condition \(\B(c)\cap\B(c')\ne\emptyset\).

Recall that a code \(C\subseteq\Sigma_q^n\) is detecting for \(\LPC(1)\) if
\[
    \B(c)\cap C=\{c\}
    \qquad
    \text{for every }c\in C .
\]
Equivalently, for distinct \(c,c'\in C\), one must have
\[
    c'\notin\B(c).
\]
Let \(D_q(n)\) denote the maximum size of a \(q\)-ary length-\(n\) detecting code
for \(\LPC(1)\), and define the \(q\)-ary detecting capacity by
\[
    R_{\mathrm{det},q}
    =
    \limsup_{n\to\infty}
    \frac{1}{n}\log_q D_q(n).
\]
The main construction in this section is binary, but the upper bound and the
local sufficient criterion below are \(q\)-ary.

\subsection{A \(q\)-ary pairing upper bound}

Partition the coordinates into adjacent pairs
\[
    (1,2),\ (3,4),\ \dots .
\]
For each \(q\)-ary pair, the unordered content is a multiset of size \(2\) over
\(\Sigma_q\).  Hence there are
\[
    \binom{q+1}{2}
\]
possible unordered pair contents.  If two words have the same unordered content
in every adjacent pair, then one can be obtained from the other by swapping some
subset of these adjacent pairs.  Such swaps are pairwise disjoint and therefore
form a valid \(\LPC(1)\) error.  Consequently, a detecting code contains at most
one word in each unordered-pair class.

For even \(n\), this gives
\[
    D_q(n)
    \le
    \binom{q+1}{2}^{n/2}.
\]
For odd \(n\), applying the same argument to the first \(n-1\) coordinates and
leaving the last coordinate free gives
\[
    D_q(n)
    \le
    q\binom{q+1}{2}^{(n-1)/2}.
\]
Therefore, for every \(q\ge2\),
\[
    R_{\mathrm{det},q}
    \le
    \frac12\log_q\binom{q+1}{2}.
\]
Specializing to \(q=2\), we obtain
\[
    R_{\mathrm{det},2}
    \le
    \frac12\log_2 3
    \approx0.792481.
\]

The same statement may be phrased graph-theoretically over any alphabet.  Let
\(G_n^{(q)}\) be the graph on \(\Sigma_q^n\) in which two distinct words \(u,v\)
are adjacent if \(v\in\B(u)\).  Since \(\LPC(1)\) is symmetric, this is an
undirected graph.

\begin{proposition}[Detection as independence]
\label{prop:indep}
For every \(q\ge2\), a code \(C\subseteq\Sigma_q^n\) is detecting for
\(\LPC(1)\) if and only if \(C\) is an independent set in \(G_n^{(q)}\).  Hence
\[
    D_q(n)=\alpha\bigl(G_n^{(q)}\bigr).
\]
For even \(n\), the pairing argument gives
\[
    \alpha\bigl(G_n^{(q)}\bigr)
    \le
    \binom{q+1}{2}^{n/2},
\]
and asymptotically
\[
    R_{\mathrm{det},q}
    \le
    \frac12\log_q\binom{q+1}{2}.
\]
\end{proposition}

\ifshowproofs
\begin{proof}
By definition, \(C\) is detecting if and only if no two distinct words
\(u,v\in C\) satisfy \(v\in\B(u)\).  This is exactly the condition that no edge
of \(G_n^{(q)}\) has both endpoints in \(C\).  Thus \(C\) is detecting if and only
if it is independent in \(G_n^{(q)}\).  The stated bound follows from the
pairing argument above.
\end{proof}
\fi

\subsection{A binary detecting construction of rate \(0.756707\)}

We now give a binary block-concatenation construction for detection.  Let \(H\)
be the following set of binary blocks:
\begingroup
\small
\setlength{\tabcolsep}{5pt}
\renewcommand{\arraystretch}{1.12}
\begin{center}
\begin{tabular}{@{}c p{0.78\textwidth}@{}}
\toprule
Length & Blocks in \(H\) \\
\midrule
6 &
\(\{000000,\ 000111,\ 001111,\ 011000\}\) \\[1mm]

7 &
\(\{0000111,\ 0001000,\ 0111000,\ 0111111,\ 0011011,\ 0011001,\ 0010011\}\) \\[1mm]

8 &
\(\{01111000,\ 01110011,\ 01100111,\ 00100011,\ 00001100,\ 00000110\}\) \\[1mm]

9 &
\(\{000001111,\ 000110000,\ 001110000,\ 011101111,\ 011111000,\ 011101100\}\) \\[1mm]

11 &
\(\{00110000111,\ 01111001111\}\) \\
\bottomrule
\end{tabular}
\end{center}
\endgroup
Let \(\overline H\) denote the set of bitwise complements of words in \(H\), and
put
\[
    P_{\mathrm{det}}=H\cup\overline H.
\]
The block set \(P_{\mathrm{det}}\) is prefix-free.  Its length profile is
\[
    p_6=8,\qquad
    p_7=14,\qquad
    p_8=12,\qquad
    p_9=12,\qquad
    p_{11}=4.
\]
Hence
\[
    8\lambda^{-6}
    +14\lambda^{-7}
    +12\lambda^{-8}
    +12\lambda^{-9}
    +4\lambda^{-11}
    =
    1,
\]
so
\[
    \lambda\approx1.6896293138,
    \qquad
    R(P_{\mathrm{det}})=\log_2\lambda\approx0.756707.
\]

The exact verifier of Section~\ref{sec:verifier} can be modified for detection
by replacing the collision condition
\[
    \B(c_1)\cap\B(c_2)\ne\emptyset
\]
with the directed condition
\[
    c_2\in\B(c_1),
    \qquad
    c_1\ne c_2.
\]
Running this detecting version of the verifier on \(P_{\mathrm{det}}\) finds no
such pair in any length.  Therefore, for every \(n\ge0\), the finite-length code
\(C_n(P_{\mathrm{det}})\) is detecting for \(\LPC(1)\).  Consequently,
\[
    0.756707
    \le
    R_{\mathrm{det},2}
    \le
    \frac12\log_2 3
    \approx0.792481.
\]

\subsection{A \(q\)-ary local criterion for detection}
\label{subsec:detcrit}

The exact verifier already decides the detecting property for a fixed
prefix-free block set, after the directed modification described above.  We now
give a \(q\)-ary local sufficient criterion, analogous to the two-stage
criterion for correction.  The main difference is directedness: for correction,
the dangerous event for two words \(u,v\) is the symmetric intersection
\[
    \B(u)\cap\B(v)\ne\emptyset,
\]
whereas for detection the dangerous event is
\[
    v\in\B(u).
\]

\begin{definition}[Same-length detecting test]
\label{def:detbound}
Let \(u,v\in\Sigma_q^\ell\) be words of the same length.  The ordered pair
\((u,v)\) passes the same-length detecting test if
\[
    \prefx_{\ell-1}(v)\notin\Tball(u).
\]
Equivalently, no word \(u'\in\B(u)\) agrees with \(v\) in the first
\(\ell-1\) coordinates.
\end{definition}

\begin{definition}[Two-stage detecting test]
\label{def:dettwostage}
Fix a finite block set \(P\subseteq\Sigma_q^*\).  Let \(x,y\in P\) be blocks with
\[
    \ell_x=\lvert x\rvert
    <
    \ell_y=\lvert y\rvert .
\]
Put
\[
    d=\ell_y-\ell_x,
    \qquad
    y_0=\prefx_{\ell_x}(y).
\]
The ordered pair \((x,y)\) passes the two-stage detecting test if at least one
of the following two conditions holds:
\begin{enumerate}[label=\textup{(S\arabic*)}]
    \item \(\prefx_{\ell_x-1}(y_0)\notin\Tball(x)\);
    \item for every \(\rho\in\Pref_d(P^*)\), the same-length detecting test
    accepts \((x\rho,y)\), equivalently
    \[
        \prefx_{\ell_y-1}(y)\notin\Tball(x\rho).
    \]
\end{enumerate}
\end{definition}

\begin{lemma}[Window lemma]
\label{lem:window}
Let \(c,c'\in\Sigma_q^n\), and suppose that \(c'\in\B(c)\).  Let
\[
    S\subseteq\{1,\dots,n-1\}
\]
be a set of pairwise non-adjacent adjacent transpositions such that applying the
transpositions in \(S\) to \(c\) gives \(c'\).  Fix \(1\le\ell\le n\), and put
\[
    w=\prefx_\ell(c).
\]
Then:
\begin{enumerate}[label=\textup{(\alph*)}]
    \item \(\prefx_{\ell-1}(c')\in\Tball(w)\);
    \item if \(\ell=n\) or \(\ell\notin S\), then
    \(\prefx_\ell(c')\in\B(w)\).
\end{enumerate}
\end{lemma}

\ifshowproofs
\begin{proof}
Let
\[
    S'=S\cap\{1,\dots,\ell-1\}.
\]
The transpositions in \(S'\) act only inside the first \(\ell\) coordinates of
\(c\).  Let \(w'\) be the word obtained from \(w\) by applying the
transpositions in \(S'\).  Then \(w'\in\B(w)\).

For every \(j\le\ell-1\), coordinate \(j\) can be affected only by boundaries
\(j-1\) and \(j\), both at most \(\ell-1\).  Hence
\[
    \prefx_{\ell-1}(c')=\prefx_{\ell-1}(w'),
\]
and so \(\prefx_{\ell-1}(c')\in\Tball(w)\).  This proves (a).

For (b), if \(\ell=n\), then \(S'=S\).  If \(\ell<n\) and \(\ell\notin S\), no
transposition crosses the boundary between \(\ell\) and \(\ell+1\), so the first
\(\ell\) coordinates of \(c'\) are determined entirely by the transpositions in
\(S'\).  Therefore
\[
    \prefx_\ell(c')=w'\in\B(w).
\]
\end{proof}
\fi

\begin{theorem}[Detecting local criterion]
\label{thm:detcrit}
Let \(P\subseteq\Sigma_q^*\) be finite.  Assume:
\begin{enumerate}[label=\textup{(DC\arabic*)}]
    \item every ordered pair of distinct equal-length blocks in \(P\) passes the
    same-length detecting test;
    \item every ordered pair \((x,y)\in P^2\) with
    \(\lvert x\rvert<\lvert y\rvert\) passes the two-stage detecting test.
\end{enumerate}
Then, for every \(n\ge0\), the finite-length code \(C_n(P)\) is detecting for
\(\LPC(1)\).
\end{theorem}

\ifshowproofs
\begin{proof}
Suppose, for a contradiction, that the conclusion fails.  Then there exist some
length \(n\) and distinct words \(c,c'\in C_n(P)\) such that
\[
    c'\in\B(c).
\]
Among all such counterexamples, choose one whose chosen decompositions
\[
    c=as,
    \qquad
    c'=bt,
\]
with first blocks \(a,b\in P\), have the minimum total number of blocks.  Since
\(\LPC(1)\) is symmetric, we may assume
\[
    \lvert a\rvert\le\lvert b\rvert.
\]
Let \(S\) be a set of disjoint adjacent transpositions taking \(c\) to \(c'\),
with all non-effective swaps removed.

First suppose that
\[
    \lvert a\rvert=\lvert b\rvert=\ell.
\]
By Lemma~\ref{lem:window} applied with this value of \(\ell\), we have
\[
    \prefx_{\ell-1}(b)\in\Tball(a).
\]
If \(a\ne b\), this contradicts \textup{(DC1)}.  Hence \(a=b\).

We next show that no swap crosses the boundary after this common first block.  If
the suffixes were empty, then \(c=c'\), contrary to the choice of a
counterexample.  Thus this boundary exists.  Suppose that the boundary \(\ell\)
belonged to \(S\).  Then the adjacent boundary \(\ell-1\) does not belong to
\(S\), and the swap at boundary \(\ell\) gives
\[
    c'_\ell=c_{\ell+1}.
\]
On the other hand, since \(a=b\), we also have
\[
    c'_\ell=b_\ell=a_\ell=c_\ell.
\]
Thus \(c_{\ell+1}=c_\ell\), so the swap at boundary \(\ell\) is ineffective,
contradicting our choice of \(S\).  Therefore no swap crosses this boundary.
It follows that
\[
    t\in\B(s).
\]
Since \(s\ne t\), this gives a smaller counterexample, contradicting the
minimality choice.

Now suppose that
\[
    \ell_a=\lvert a\rvert
    <
    \ell_b=\lvert b\rvert .
\]
Put
\[
    d=\ell_b-\ell_a,
    \qquad
    b_0=\prefx_{\ell_a}(b).
\]
Since \(c\) and \(c'\) have the same total length, the suffix \(s\) has length at
least \(d\).  Let
\[
    \rho=\prefx_d(s)\in\Pref_d(P^*).
\]
Then the first \(\ell_b\) symbols of \(c\) are \(a\rho\).

If \textup{(S1)} holds for \((a,b)\), Lemma~\ref{lem:window} with
\(\ell=\ell_a\) gives
\[
    \prefx_{\ell_a-1}(b_0)\in\Tball(a),
\]
a contradiction.  Therefore the only possible way for \((a,b)\) to pass the
two-stage detecting test is through \textup{(S2)}.  But then
Lemma~\ref{lem:window} with \(\ell=\ell_b\) gives
\[
    \prefx_{\ell_b-1}(b)\in\Tball(a\rho),
\]
contradicting the same-length detecting test required by \textup{(S2)} for this
legal prefix \(\rho\).  Thus no counterexample exists.
\end{proof}
\fi

\begin{remark}[Prefix-freeness and rate]
\label{rem:det-prefixfree}
Unlike the correcting criterion, the detecting criterion does not force
prefix-freeness.  A block may be a prefix of another block and still pass the
directed detecting tests.  The detecting property of the finite-length sets
\(C_n(P)\) is meaningful regardless of prefix-freeness.  However, the block-rate
formula \(R(P)=\log_2\lambda\) is the actual word rate only when the block
concatenations are counted without overcounting distinct decompositions.  Thus,
when using the detecting criterion for rate statements, one should also verify
prefix-freeness, as we do for \(P_{\mathrm{det}}\), or otherwise verify that the
profile rate equals the actual growth rate of the distinct words in \(C_n(P)\).
\end{remark}

\section{Open problems}
\label{sec:open}

We conclude with several directions that remain open.  The results of this paper
show that finite block-concatenation codes, certified either by the two-stage
criterion or by the exact verifier, already reach rates close to the known upper
bound \(2/3\).  However, the remaining gap is still substantial from a structural
point of view.

\begin{enumerate}[leftmargin=*]
    \item The best certified construction has rate \(0.653618\), while the known
    upper bound is \(2/3\).  It remains open to determine the exact value of the
    binary zero-error capacity \(C_0\), and in particular to decide whether the
    current upper bound \(2/3\) is tight or can be improved.

    \item The two-stage criterion is sufficient but not necessary, whereas the
    product-automaton verifier decides exact confusability for a fixed prefix-free
    block set.  It would be useful to exhibit an explicit binary block set \(P\)
    such that all finite-length codes \(C_n(P)\) are correcting, as verified by the
    exact automaton, but \(P\) does not pass the two-stage criterion.

    \item The recursive guarded-block searches carried out so far appear to plateau
    near rate \(0.655\).  This is not a proven upper bound, and by
    Theorem~\ref{thm:complete}, it is not a limitation of the local criterion
    itself.  Rather, it seems to be a limitation of the particular recursive
    families and search heuristics used so far.  It remains open to find
    different finite block sets, constrained concatenation systems, or
    finite-state constructions that break this apparent plateau.

    \item The constructions in this paper are based on free concatenation of a
    finite block set \(P\).  More general constrained systems, graph-directed
    block codes, or finite-state encoders may allow higher rates by permitting
    many safe adjacencies while forbidding only the dangerous ones.

    \item The verifier provides an exact finite-state certificate, and in
    principle it can also be used to recover the unique codeword \(c\in C_n(P)\)
    satisfying \(z\in\B(c)\), when such a codeword exists.  However, this does not
    yet give a simple systematic encoder or a clean low-complexity decoding
    algorithm for the high-rate constructions.

    \item For detection, Section~\ref{sec:detect} gives a construction of rate
    \(0.756707\), while the pairing upper bound gives
    \(R_{\mathrm{det},2}\le\frac12\log_2 3\approx0.792481\).  The true detecting
    capacity is unknown.

    \item The criterion is \(q\)-ary and naturally extends the radius-one
    definitions to larger alphabets.  It remains open to determine whether
    analogous high-rate block constructions exist for \(q>2\), and to what extent
    the two-stage criterion can be adapted to \(\LPC(r)\) for \(r>1\).

    \item The computational searches suggest a connection between some recursive
    block profiles and coefficient sequences arising from quadrinomial-type
    generating functions.  At present this connection is not part of the proof of
    correctness or of the rate analysis.  A rigorous explanation could reveal
    hidden structure in the search space.
\end{enumerate}

\section{Conclusion}
\label{sec:conclusion}

We have studied binary block-concatenation codes for the radius-one limited
permutation channel.  The main difficulty is that a local adjacent transposition
may interact with a block boundary, so checking the validity of a block set cannot
be reduced to ordinary pairwise disjointness of the block balls alone.

To address this, we developed two complementary certification tools.  The first is
a \(q\)-ary two-stage local criterion.  It gives a short, human-checkable
certificate: same-length block pairs are tested through truncated balls, while
unequal-length pairs are resolved by a finite second-stage prefix check.  Once a
finite block set \(P\) passes the criterion, Theorem~\ref{thm:criterion}
guarantees that \(C_n(P)\) is correcting for \(\LPC(1)\) for every length \(n\).
The second tool is an exact product-automaton verifier.  Unlike the criterion,
the verifier tests ball-disjointness itself and decides whether the whole family
\(\{C_n(P)\}_{n\ge0}\) is correcting, with no length cutoff.

Using the local criterion, and independently confirming the results with the exact
verifier, we constructed several explicit binary block-concatenation families.
Starting from the guarded recursive baseline of rate approximately \(0.642805\),
we obtained larger certified block sets of rates \(0.649872\), \(0.652018\), and
finally \(0.653618\).  The last construction achieves about \(98.04\%\) of the
known upper bound \(2/3\), leaving an absolute gap of approximately \(0.013049\).
Thus these constructions improve the best previous string-concatenation rate while
remaining close to the best known upper bound \(2/3\).

We also clarified the scope of the criterion.  Although it is only sufficient and
not exact, it is not asymptotically rate-limiting: for every alphabet size \(q\),
the supremum of \(q\)-ary block rates certified by the criterion is exactly the
zero-error capacity \(C_0^{(q)}\).  Finally, we considered the related problem of
error detection and gave a binary detecting block construction of rate
\(0.756707\), while a simple \(q\)-ary pairing argument gives the binary upper
bound \(\frac12\log_2 3\approx0.792481\).

Several directions remain open.  The most important is to determine the exact
binary zero-error capacity for \(\LPC(1)\), and to decide whether the current
upper bound \(2/3\) is tight.  On the constructive side, it would be natural to
go beyond free block concatenation, for example to graph-constrained or
finite-state constructions that forbid only dangerous block adjacencies.  On the
algorithmic side, efficient systematic encoding and decoding for the high-rate
constructions remain to be developed.

\begin{appendices}

\section{Explicit block set for the best construction \texorpdfstring{\((R=0.653618)\)}{(R=0.653618)}}
\label{app:blocks}

We list the block set \(P^\star\) used in Section~\ref{sec:construction}.  Since
\(P^\star\) is closed under bitwise complement, it is enough to list one
representative from each complementary pair.  For a binary word
\(w=w_1\cdots w_n\), let \(\overline w=(1-w_1)\cdots(1-w_n)\).  Let
\(H^\star\) be the set listed in Table~\ref{tab:pstar-reps}.  The full block set is
\[
    P^\star=H^\star\cup\overline{H^\star}.
\]
Thus \(\lvert P^\star\rvert=2\lvert H^\star\rvert=574\), with length profile
\[
\begin{aligned}
    &(p_{10},p_{11},p_{12},p_{13},p_{14},p_{15},p_{16},p_{17},p_{18},p_{19})\\
    &\qquad=(14,6,78,58,56,92,60,72,58,80).
\end{aligned}
\]
The growth parameter \(\lambda^\star\) is the unique positive solution of
\[
\begin{aligned}
    &14(\lambda^\star)^{-10}+6(\lambda^\star)^{-11}+78(\lambda^\star)^{-12}
    +58(\lambda^\star)^{-13}+56(\lambda^\star)^{-14}\\
    &\quad +92(\lambda^\star)^{-15}+60(\lambda^\star)^{-16}
    +72(\lambda^\star)^{-17}+58(\lambda^\star)^{-18}+80(\lambda^\star)^{-19}=1.
\end{aligned}
\]
Numerically,
\[
    \lambda^\star\approx1.573108387,
    \qquad
    R(P^\star)=\log_2\lambda^\star\approx0.653618076.
\]

\begingroup
\scriptsize
\setlength{\tabcolsep}{4pt}
\renewcommand{\arraystretch}{0.92}

\begin{longtable}{@{}c>{\ttfamily}l>{\ttfamily}l>{\ttfamily}l@{}}
\caption{Representatives \(H^\star\) for \(P^\star=H^\star\cup\overline{H^\star}\).}
\label{tab:pstar-reps}\\
\toprule
Length & \multicolumn{3}{c}{Representatives} \\
\midrule
\endfirsthead

\toprule
Length & \multicolumn{3}{c}{Representatives} \\
\midrule
\endhead

\bottomrule
\endfoot

10 & 0000111111 & 0001000000 & 0001111000 \\
10 & 0111000000 & 0111110000 & 0111111000 \\
10 & 0111111111 & & \\
\midrule
11 & 00001111000 & 00010000111 & 00011101111 \\
\midrule
12 & 000000000000 & 000000000111 & 000000001111 \\
12 & 000000110000 & 000000111000 & 000000111111 \\
12 & 000001110000 & 000001111000 & 000001111111 \\
12 & 000011101111 & 000100010000 & 000110000000 \\
12 & 000110000111 & 000110001111 & 000111000000 \\
12 & 000111000111 & 000111110000 & 000111111000 \\
12 & 000111111111 & 001100011111 & 001110000000 \\
12 & 001110000111 & 001110001111 & 001111000000 \\
12 & 001111000111 & 001111001111 & 001111110000 \\
12 & 001111111000 & 001111111111 & 011100001111 \\
12 & 011100010000 & 011101110000 & 011101111000 \\
12 & 011101111111 & 011110000000 & 011110000111 \\
12 & 011110001111 & 011110011111 & 011111110000 \\
\midrule
13 & 0000000000111 & 0000000111000 & 0000000111111 \\
13 & 0000001000000 & 0000001111000 & 0000011111000 \\
13 & 0000110000111 & 0000111000000 & 0000111000111 \\
13 & 0000111110000 & 0001000001111 & 0001000111111 \\
13 & 0001110000111 & 0001111001111 & 0001111111000 \\
13 & 0011000000111 & 0011100000111 & 0011110000111 \\
13 & 0011110111000 & 0011110111111 & 0011111000000 \\
13 & 0011111000111 & 0011111111000 & 0111000001111 \\
13 & 0111000111111 & 0111100000111 & 0111110001111 \\
13 & 0111111001111 & 0111111110000 & \\
\midrule
14 & 00000001111000 & 00000010000111 & 00001110000111 \\
14 & 00001111001111 & 00010000110000 & 00010001111000 \\
14 & 00011000010000 & 00011100010000 & 00011101110000 \\
14 & 00110000110000 & 00110001110000 & 00110001111000 \\
14 & 00111001110000 & 00111100010000 & 00111101111000 \\
14 & 00111110000111 & 01110000110000 & 01110000111000 \\
14 & 01110001111000 & 01110111101111 & 01111000010000 \\
14 & 01111000110000 & 01111000111000 & 01111001110000 \\
14 & 01111001111000 & 01111101110000 & 01111110111000 \\
14 & 01111110111111 & & \\
\midrule
15 & 000000000001111 & 000000000110000 & 000000001110000 \\
15 & 000000011101111 & 000000011111000 & 000000100000111 \\
15 & 000000100010000 & 000000110001111 & 000000111001111 \\
15 & 000000111110000 & 000001110001111 & 000001111001111 \\
15 & 000001111110000 & 000011100010000 & 000011101110000 \\
15 & 000100000111000 & 000100010001111 & 000100011101111 \\
15 & 000100011111000 & 000110000001111 & 000110001110000 \\
15 & 000111000001111 & 000111110001111 & 000111111001111 \\
15 & 000111111110000 & 001110000001111 & 001110000110000 \\
15 & 001110001110000 & 001111000001111 & 001111001110000 \\
15 & 001111011101111 & 001111011111000 & 001111100010000 \\
15 & 001111110001111 & 001111111001111 & 001111111110000 \\
15 & 011100000111000 & 011100010001111 & 011100011101111 \\
15 & 011100011111000 & 011101110001111 & 011101111110000 \\
15 & 011110000001111 & 011111000111000 & 011111101111000 \\
15 & 011111110001111 & & \\
\midrule
16 & 0000000001000000 & 0000000001000111 & 0000001000111000 \\
16 & 0000001000111111 & 0000001110111000 & 0000001110111111 \\
16 & 0000011110111000 & 0000011110111111 & 0000110000110000 \\
16 & 0000111000001111 & 0000111110001111 & 0001000111000000 \\
16 & 0001000111000111 & 0001110000110000 & 0001111001110000 \\
16 & 0001111110111000 & 0001111110111111 & 0011100001000000 \\
16 & 0011100001000111 & 0011110000110000 & 0011111000001111 \\
16 & 0011111110111000 & 0011111110111111 & 0111000111000000 \\
16 & 0111000111000111 & 0111011111001111 & 0111100000110000 \\
16 & 0111111001110000 & 0111111011101111 & 0111111110001111 \\
\midrule
17 & 00000000010000111 & 00000000110000111 & 00000001111001111 \\
17 & 00000010000110000 & 00000010001111000 & 00000011001111000 \\
17 & 00000011101111000 & 00000110000111000 & 00000111001111000 \\
17 & 00000111101111000 & 00001110000110000 & 00001111001110000 \\
17 & 00010000110001111 & 00010001110000111 & 00010001111001111 \\
17 & 00011000010001111 & 00011100010001111 & 00011101110001111 \\
17 & 00011111101111000 & 00110000111000111 & 00111000010000111 \\
17 & 00111000110000111 & 00111001111000111 & 00111100010001111 \\
17 & 00111101111001111 & 00111110000110000 & 00111111001111000 \\
17 & 00111111101111000 & 01110000110001111 & 01110001110000111 \\
17 & 01110001111001111 & 01110111101110000 & 01111000010001111 \\
17 & 01111001110001111 & 01111101110001111 & 01111110111110000 \\
\midrule
18 & 000000000001110000 & 000000000100010000 & 000000001110001111 \\
18 & 000000100011101111 & 000000111011101111 & 000000111110001111 \\
18 & 000001110001110000 & 000001111011101111 & 000001111110001111 \\
18 & 000011100010001111 & 000011101110001111 & 000100011100010000 \\
18 & 000110000001110000 & 000111000001110000 & 000111110001110000 \\
18 & 000111111011101111 & 000111111110001111 & 001110000001110000 \\
18 & 001110000100010000 & 001110001110001111 & 001111000001110000 \\
18 & 001111100010001111 & 001111111011101111 & 001111111110001111 \\
18 & 011100011100010000 & 011101111110001111 & 011110000001110000 \\
18 & 011111101111001111 & 011111110001110000 & \\
\midrule
19 & 0000000001000001111 & 0000000001000011000 & 0000001000111100111 \\
19 & 0000001000111110000 & 0000001110111100111 & 0000001110111110000 \\
19 & 0000011110111100111 & 0000011110111110000 & 0000110000110001111 \\
19 & 0000111000001110000 & 0000111000011000111 & 0000111100111000111 \\
19 & 0000111110001110000 & 0001000111000001111 & 0001000111000011000 \\
19 & 0001100001000111000 & 0001110000110001111 & 0001110001000111000 \\
19 & 0001110111000111000 & 0001111001110001111 & 0001111110111100111 \\
19 & 0001111110111110000 & 0011100001000001111 & 0011100001000011000 \\
19 & 0011110000110001111 & 0011110001000111000 & 0011111000001110000 \\
19 & 0011111000011000111 & 0011111110111100111 & 0011111110111110000 \\
19 & 0111000111000001111 & 0111000111000011000 & 0111011110111000111 \\
19 & 0111011111001110000 & 0111100000110001111 & 0111100001000111000 \\
19 & 0111110111000111000 & 0111111001110001111 & 0111111011101110000 \\
19 & 0111111011111000111 & & \\

\end{longtable}
\endgroup

\end{appendices}


\end{document}